\documentclass[aps,pra,superscriptaddress,nofootinbib,twocolumn]{revtex4-2}
\usepackage{amsmath}
\usepackage{amsfonts}
\usepackage{amssymb}
\usepackage{amsthm}
\usepackage{mathrsfs}
\usepackage{multirow}
\usepackage[inline]{enumitem}
\usepackage[colorlinks=true,urlcolor=blue,citecolor=cyan,linkcolor=blue]{hyperref}
\usepackage{xcolor}

\newcommand{\bra}[1]{\langle #1 \rvert}
\newcommand{\ket}[1]{\left\lvert #1 \right\rangle}

\newcommand{\lmap}[1]{\mathcal{L}\left(#1\right)}
\newcommand{\hil}{\mathcal{H}}
\newcommand{\hils}[1]{\mathcal{H}_\mathrm{#1}}
\newcommand{\Pis}[1]{\Pi_\mathrm{#1}}

\newtheorem{proposition}{Proposition}
\newtheorem{lemma}{Lemma}
\newtheorem{corollary}{Corollary}
\newtheorem{theorem}{Theorem}

\theoremstyle{definition}

\begin{document}
\title{Existence of Pauli-like stabilizers for every quantum error-correcting code}
\author{Jhih-Yuan Kao} \email{frankkao@ntu.edu.tw}
\affiliation{Department of Physics and Center for Theoretical Physics, National Taiwan University, Taipei 106319, Taiwan}
\affiliation{Center for Quantum Science and Engineering, National Taiwan University, Taipei 106319, Taiwan}
\author{Hsi-Sheng Goan} \email{goan@phys.ntu.edu.tw}
\affiliation{Department of Physics and Center for Theoretical Physics, National Taiwan University, Taipei 106319, Taiwan}
\affiliation{Center for Quantum Science and Engineering, National Taiwan University, Taipei 106319, Taiwan}
\affiliation{Physics Division, National Center for Theoretical Sciences, Taipei, 106319, Taiwan}

\begin{abstract}
    The Pauli stabilizer formalism is perhaps the most thoroughly studied means of procuring quantum error-correcting codes, whereby the code is obtained through commutative Pauli operators and ``stabilized'' by them. In this work we will show that every quantum error-correcting code, including Pauli stabilizer codes and subsystem codes, has a similar structure, in that the code can be stabilized by commutative ``Paulian'' operators which share many features with Pauli operators and which form a \textbf{Paulian stabilizer group}. By facilitating a controlled gate we can measure these Paulian operators to acquire the error syndrome. Examples concerning codeword stabilized codes and bosonic codes will be presented; specifically, one of the examples has been demonstrated experimentally and the observable for detecting the error turns out to be Paulian, thereby showing the potential utility of this approach. This work provides a possible approach to implement error-correcting codes and to find new codes.
\end{abstract}
\maketitle

\section{Introduction}

Quantum information is stored as quantum states. Due to defects in the devices or executions, and the inevitable interaction of the quantum system with the environment, the state of the quantum system can be changed in a nondeterministic manner, which is an error; consequently, error correction is vital for the information to stay hygienic. Using quantum error-correcting codes, states are prepared in specific subspaces such that if certain errors occur, we can detect and correct them \cite{Gottesman97,*Gottesman97arXiv,Gottesman10,*Gottesman09arXiv,Terhal15,Nielsen,Roffe19}. Even though quantum devices without error correction may serve certain purposes such as simulating physical systems \cite{Buluta09,Georgescu14}, a universal quantum computer that is scalable still requires error correction \cite{Preskill18,Cirac21}.

Pauli stabilizer codes \cite{Gottesman96,Gottesman97,Calderbank98} are an extremely important class of quantum error-correcting codes. Some of the most promising codes, such as topological codes \cite{Kitaev03,Dennis02,Bombin06,Raussendorf07PRL,Fowler12PRX} which include surface codes \cite{Bravyi98,Freedman01,Dennis02,Horsman12,Fowler12,Hill15,Versluis17,Landau16,Takita16,Yoder17,Ataides21}, and quantum low-density parity-check (LDPC) codes \cite{Camara07,Gottesman13arXiv,Baber15,Breuckmann21}, are based on Pauli stabilizer codes. An advantage of Pauli stabilizer formalism is that it informs us of which measurements to implement to detect the errors, namely the stabilizer generators.

There are several ways of generalizing the Pauli stabilizer formalism, for example, by generalizing Pauli groups, or nice error bases to nonbinary cases \cite{Knill96_1,*Knill96_1arXiv,Knill96_2,*Knill96_2arXiv,Ashikhmin01,Ketkar06,Nadkarni21}, or by considering noncommutative groups on binary codes \cite{Ni15}. In this work, instead of defining a certain group and constructing an error-correcting code from it, we will do the opposite: We investigate the structure of any error-correcting code, including subsystem codes \cite{Kribs05,Poulin05,Kribs06,*Kribs06arXiv,Bacon06,Aly06,*Aly06arXiv,Aliferis07,Higgott21}, to show that every code can be stabilized by a ``Paulian'' stabilizer group  (Proposition~\ref{pro:main} and Corollary~\ref{cor:sub}), the exact meaning of being Paulian to be explained in Sec.~\ref{sec:pau}. Identifying the Paulian stabilizer group of an error-correcting code may give us a guideline on how to implement such a code: The error syndrome can be obtained by measuring these Paulian operators, which can be conducted via controlled operations (Sec.~\ref{sec:cnot}). We will also show how to obtain the Paulian stabilizer group for a concatenated binary code (Sec.~\ref{sec:con}) \cite{Knill96,Gottesman97,Gottesman10}, and in Sec.~\ref{sec:ex} we will demonstrate some examples. For conciseness, details of some topics can be found in the appendixes.

\section{Preliminaries}
\label{sec:pre}
$\mathbb{A}\subseteq \mathbb{B}$ means $\mathbb{A}$ is a subset of $\mathbb{B}$, while $\subset$ indicates it is a proper subset. A map $f: \mathbb{X}\rightarrow \mathbb{Y}$ to the restriction of $\mathbb{X}'\subseteq \mathbb{X}$, denoted by $f|_{\mathbb{X}'}$, is a map from $\mathbb{X}'$ to $\mathbb{Y}$ with $f|_{\mathbb{X}'} (x) = f(x)$ $\forall x\in \mathbb{X}'$ \cite{Roman,Loomis,Kaophd,*KaophdarXiv}, for which we will often shrink the codomain to the image $f|_{\mathbb{X}'}(\mathbb{X}') = f(\mathbb{X}')$. The \emph{span} of a set of vectors is the set of all linear combinations thereof, which is a subspace. We will use shorthand to label sets obtained from others in a sensible way, e.g. $\hil^{\otimes 3}$ is $\hil\otimes\hil\otimes\hil$. The subscript beside an identity operator, denoted by $I$, or orthogonal projection, denoted by $\Pi$, indicates the (sub)space the operator acts on or projects onto; e.g. $\Pi_\text{C}$ projects onto $\hil_\text{C}$. 

The code space $\hils{C}$ of a quantum error-correcting code is a subspace of the entire space $\hil$ where the encoded state is stored \cite{Gottesman97,Preskill,Nielsen}; sometimes we simply refer to the code space as the code. With $\mathbb{C}^{n}$ denoting a generic $n$-dimensional complex vector space, a code is called an $[[n,k]]$-code if $\hil \cong \mathbb{C}^{2^n}$ and $\hil_\text{C} \cong \mathbb{C}^{2^k}$ for some integers $n$ and $k$, where $A \cong B$ indicates that $A$ and $B$ are isomorphic; such codes are said to be \textbf{binary}---We use the term binary codes in a stricter sense than, e.g., Ref.~\cite{Chen08PRA}, as we require the code space to be binary too. Also, an $((n,k,d))$-code has $n$ qubits, a code space of dimension $k$ and distance $d$ \cite{Cross09}. For a qubit system, $\ket{\pm 1}$ instead of $\ket{0}$ and $\ket{1}$ will denote the $\pm 1$-eigenstates. 

An operator is said to \emph{stabilize} a subspace $\hil'$ if $\hil'$ is a subspace of the operator's $1$-eigenspace. We will refer to the subspace spanned by all simultaneous eigenvectors with the same simultaneous eigenvalues as a \emph{simultaneous eigenspace}. $\mathsf{P}^n$ will denote the Pauli group on $(\mathbb{C}^{2})^{\otimes n} \cong \mathbb{C}^{2^n}$, and its members will be called \emph{Pauli operators}  \cite{Gottesman97,Knill96_1,Knill96_2}; in this work we will use $X_i$, $Y_i$ and $Z_i$ to denote Pauli $X$, $Y$, $Z$ operators on the $i$-th site. If the code space of a code is the $(1,\dotsc,1)$-simultaneous eigenspace of commutative Pauli operators, the code is called a \emph{Pauli stabilizer code}, and the abelian group generated by these operators is the \emph{stabilizer group} \cite{Gottesman96,Gottesman97,Calderbank98}.

A representation of a group $\mathsf{G}$ on a space $\mathcal{V}$ is a homomorphism $\Phi$ from $\mathsf{G}$ to the general linear group of $\mathcal{V}$, and it is said to be faithful if $\Phi$ is one-to-one \cite{Hall}. Abusing the language, we will call the image $\Phi(\mathsf{G})$ ``a representation.'' Two representations $\mathsf{G}_1$ on $\hil_1$ and $\mathsf{G}_2$ on $\hil_2$ of $\mathsf{G}$ are said to be unitarily equivalent if there exits a unitary map $V:\hil_1\rightarrow \hil_2$ such that $\mathsf{G}_1 = V^{-1} \mathsf{G}_2 V$ \cite{Landsman,Putnam19lecture,Blackadar}.

\subsection{Involutions} \label{sec:inv}
An operator is said to be an involution if it is its own inverse i.e., if it squares to $I$ \cite{Roman}; for instance, Pauli $X$, $Y$, $Z$ are all involutions. By definition, the spectrum of an involution can only contain $\pm 1$, which by the spectral theorem leads to
\begin{lemma}\label{lem:invuni}
    An involution on a Hilbert space is normal if and only if it is self-adjoint and if and only if it is unitary.
\end{lemma}
Self-adjoint involutions are of great physical interest, because they correspond to both physical observables (self-adjoint) and evolution of a system (unitary). A Pauli group is composed of unitary involutions and operators that square to $-I$, which we call counterinvolutions. We can easily see that a counterinvolution is an involution multiplied by $i$, and vice versa.

If a pair of involutions or counterinvolutions $A$ and $B$ anticommute, for an $a$-eigenvector $\ket{v}$ of $A$, $BA\ket{v}= aB\ket{v} = -AB\ket{v}$, and since they are by definition automorphisms, $B\ket{v}\neq 0$ for all nonzero $\ket{v}$; therefore, $B$ maps the $a$-eigenspace of $A$ to the $-a$-eigenspace, and the $\pm a$-eigenspaces are thus isomorphic.

\subsection{Paulian operators} \label{sec:pau}

An operator will be called \textbf{Paulian} if 
\begin{enumerate}
    \item it is either an involution or counterinvolution;
    \item it is unitary; and
    \item it has two isomorphic eigenspaces unless it has a single eigenspace.
\end{enumerate}
Accordingly, all Pauli operators are Paulian. A Paulian operator is self-adjoint if and only if it is an involution, and it is skew-self-adjoint if and only if it is a counterinvolution. 

When the space is finite-dimensional, we could simply require Paulian operators, except for those proportional to $I$, to be traceless. As the eigenvalues have opposite signs, the two eigenspaces have the same dimension and hence are isomorphic. However, a unitary operator on an infinite-dimensional space is not trace class \cite{Blackadar,HallQ} and in general it does not have a well-defined trace, so we simply demand the eigenspaces be isomorphic. Having isomorphic eigenspaces, the unitary map between them will play the role of Pauli $Z$ [cf. \eqref{eq:binh} and the proof for Proposition~\ref{pro:main} (Sec.~\ref{sec:stbgp})]; besides, this makes it possible to find anticommuting Paulian operators, cf. the previous subsection.

To appreciate the significance of Paulian operators in physics, we remark 
\begin{enumerate}
    \item By Lemma~\ref{lem:invuni}, a Paulian involution is unitary and self-adjoint at the same time, so it can not only describe the evolution but also be an observable.
    \item Because a Paulian operator (except for those that are proportional to $I$) is traceless or has two isomorphic eigenspaces, very roughly speaking, if an observable has two possible outcomes, and if both outcomes are equally likely on average with all states considered, then it is Paulian. 
\end{enumerate}

Finally, in this work when we refer to an operator as Paulian, it may not necessarily be Paulian on the entire domain, but only Paulian to the restriction of a specific subspace, which subspace has to do with the errors the operator can detect or correct. This will be explained in more detail later.

\subsection{Condition for error correction}
\label{sec:conerr}
The necessary and sufficient condition for a set of errors $\mathbb{E}$ to be correctable is \cite{Knill96_1,Knill96_2}
\begin{equation}
    \Pi_\text{C} E^\dagger F \Pi_\text{C} \propto \Pi_\text{C} \;\forall E,F\in\mathbb{E}.\label{eq:errcon}
\end{equation}
There are other expressions for this condition, for example, $\Pi_\text{C} E^\dagger F \Pi_\text{C} = \alpha_{E,F} \Pi_\text{C}$ where $\alpha$ is a Hermitian matrix \cite{Nielsen,Gottesman97,Preskill,Gottesman10}, or in terms of inner product and basis \cite{Knill97,Gottesman97,Gottesman10}. It is worth mentioning that the common requirement that $\alpha$ is Hermitian is somewhat superfluous: If two operators $A$ and $B$ satisfy 
$$\Pi A^\dagger B \Pi = c \Pi$$ 
for some constant $c$ and orthogonal projection $\Pi$, then it must be true that 
$$\Pi B^\dagger A \Pi = \left( \Pi A^\dagger B \Pi \right)^\dagger = c^* \Pi.$$
Hence the matrix $\alpha$ above is naturally Hermitian. If a code can correct $\mathbb{E}$, it can correct any error in the span of $\mathbb{E}$.

From \cite{Nielsen,Preskill}, we can find a maximal subset $\mathbb{F}$ of $\text{span}\mathbb{E}$ whose elements obey
\begin{equation}
    \Pi_\text{C} E^\dagger F \Pi_\text{C} = \begin{cases}
        0, & E\neq F \\
        \Pi_\text{C}, & E = F
    \end{cases} \; \forall E,F\in\mathbb{F},\label{eq:errEi}
\end{equation}
and we call correctable errors in $\mathbb{F}$ \emph{orthonormal}; the set $\mathbb{F}$ is maximal in the sense that
\begin{equation}
    \sum_{E\in \mathbb{E}} E \hil_\text{C} = \bigoplus_{F\in \mathbb{F}} F \hil_\text{C}, \label{eq:ehc}
\end{equation}
where $\oplus$ denotes an orthogonal direct sum, is satisfied. Note 
\begin{equation}
    E \hils{C}\cong \hils{C} \;\forall E\in \mathbb{E},
\end{equation}
so $\bigoplus_{F\in \mathbb{F}} F \hil_\text{C}$ is an orthogonal direct sum of isomorphic spaces. On the other hand, if we have a set of errors or operators such that the operators in it are ``orthogonal'' but not necessarily ``normalized,'' i.e., $\Pi_\text{C} E^\dagger E \Pi_\text{C} = c_E \Pis{C}$ for some scalar $c_E$ that is not necessarily 1, then the set is referred to as \emph{orthogonal}.

\section{Paulian stabilizer group}
\label{sec:bistg}

Here is the main result of this work, which will be explained in detail soon after; $\dim$ below refers to the dimension of a vector space: 
\begin{proposition}\label{pro:main}
    Consider an error-correcting code, with the code space $\hil_\mathrm{C}$ belonging in $\hil$. There exist operators which stabilize $\hils{C}$ and satisfy the following properties:
    \begin{enumerate}
        \item To the restriction of a $2^m k'$-dimensional subspace $\hil'$ for some positive integer $m$ with $\hil_\mathrm{C}\subseteq \hil' \subseteq \hil$ and $k' \geq \dim \hils{C}$, these operators are mutually commutative Paulian operators, forming an abelian group $\mathsf{S}$ called the \textbf{Paulian stabilizer group}, which is generated by $m$ operators. If $\hil$ is infinite-dimensional, $\hil'$ can be as well.
        \item $\mathsf{S}$ is an abelian subgroup of a group of Paulian operators $\mathsf{P}^m_\mathrm{S}$, which is a faithful representation of $\mathsf{P}^m$.
        \item A subset of all correctable errors can be detected by measuring these operators and corrected by applying proper inverses. 
    \end{enumerate}
\end{proposition}

\subsection{The minimal stabilizer group} \label{sec:stbgp}
First, we will prove a ``minimal'' version of this proposition, which yields a ``minimal'' Paulian stabilizer group. The reader may skim over the proof and come back later when necessary.

\begin{proof}
    With $\mathbb{F}$ defined in \eqref{eq:errEi}, we choose a subset of $\mathbb{F}' \subseteq \mathbb{F}$ whose cardinality is a positive integral power of 2, $m$, with $I \in \mathbb{F}'$. As long as the code is nontrivial, such a subset always exists. We want $\mathbb{F}'$ to be as large as possible, so we choose
    \begin{equation}
        m = \left \lfloor \log_2 \left|\mathbb{F}\right| \right \rfloor,\label{eq:fpri}
    \end{equation}
    where $\left\lfloor \cdot \right \rfloor$ is the floor function; we thus have $|\mathbb{F}'| = 2^m$. 

    Let $\mathbb{T}$ be the set of all tuples of $\pm 1$ with length $m$ and $(t)$ be the symbol for elements in $\mathbb{T}$, which we will use for indexing. For each $F\in\mathbb{F}'$, choose a a unique tuple $(t)\in \mathbb{T}$; to put it another way, we define a bijective ``syndrome map'' $f_{\mathrm{sym}}: \mathbb{F}'\rightarrow \mathbb{T}$ such that $f_{\mathrm{sym}}(F) \in \mathbb{T}$ is the tuple corresponding to $F$, which, as we will see, is the syndrome of $F$. $F_{(t)}$ will denote the error $(t)\in \mathbb{T}$ refers to:
    \begin{equation}
        F_{(t)} := f_{\mathrm{sym}}^{-1}\left((t)\right), \label{eq:synmap}
    \end{equation}
    and likewise\footnote{Later $\hil_{(t)}$ will be defined as the $(t)$-simultaneous eigenspace of the stabilizers. Hence \eqref{eq:ht} is not the definition of $\hil_{(t)}$, but it is true here.} 
    \begin{equation}
        \hil_{(t)} = F_{(t)} \hils{C}.\label{eq:ht}
    \end{equation}
    Among all such binary tuples $(t)$, 
    \begin{equation}
        (I) := (1,\dotsc,1)
    \end{equation}
    will serve as a convenient abbreviation; in particular we require
    \begin{equation}
        F_{(I)} = I,
    \end{equation}    
    namely $f_{\mathrm{sym}}(I)$ is selected to be $(I) = (1,\dotsc,1)$.
    We also define
    \begin{equation}
        \overline{\hil} := \bigoplus_{F\in \mathbb{F}'} F \hils{C} = \bigoplus_{(t)\in \mathbb{T}} \hil_{(t)} \subseteq \hil.\label{eq:h}
    \end{equation}

    Here let $\hil'$ of this proposition be $\overline{\hil}$. With
    \begin{equation}
        \dim \overline{\hil} = 2^m \dim\hils{C},
    \end{equation}
    it means $k'$ is $\dim\hils{C}$. We have the following isomorphism:
    \begin{equation}
        \hil_\text{B} := \hil_\text{C} \otimes \bigotimes_{i=1}^m \mathbb{C}_i^2 \cong \overline{\hil},\label{eq:binh}
    \end{equation}
    where the subscript $i$ of $\mathbb{C}_i^2$ is for indexing.

    Let's construct a unitary map $U:\overline{\hil} \rightarrow\hil_\text{B}$ as follows: Since $\hil_{(t)}$'s are isomorphic, there exist unitary maps
    \begin{equation}
        V_{(t)}: \hil_{(t)} \rightarrow \hil_\text{C} \;\forall (t)\in\mathbb{T}, \label{eq:vi}
    \end{equation}
    among which we let $V_{(I)}: \hil_\text{C}\rightarrow \hil_\text{C}$ be $I_\text{C}$. Let's also choose an orthonormal basis $\{\ket{\pm 1}_i\}$ for each $\mathbb{C}_i^2$. For any $(t) = (i_1,\dotsc,i_m) \in \mathbb{T}$ and any $\ket{v} \in \mathcal{H}_{(t)}$, let 
    \begin{equation}
        U \ket{v} := \left(V_{(t)} \ket{v}\right) \otimes \ket{(t)}, \label{eq:U}
    \end{equation}
    where 
    \begin{equation}
        \ket{(t)}:= \ket{i_1}_1 \otimes \cdots \otimes\ket{i_m}_m \in \bigotimes_{i=1}^m \mathbb{C}_i^2. \label{eq:keti}
    \end{equation}
    By definition \eqref{eq:h}, $\overline{\hil}$ is the direct sum of $\hil_{(t)}$'s, so $U$ of \eqref{eq:U} is defined on the entirety of $\overline{\hil}$. $U$ is unitary because $V_{(t)}$'s are unitary and $\{\ket{\pm 1}_i\}$'s are orthonormal bases.

    Now, for every $i=1,\dotsc,m$ let $X_i$ and $Z_i$ denote the operators on $\hil_\text{B}$ that apply Pauli $X$ and $Z$ on $\mathbb{C}^2_i$ and act trivially on the other subsystems including $\hil_\text{C}$. Their counterparts on $\overline{\hil}$ via $U$ are
    \begin{equation}
        Z_i^\text{S} := U^{-1} Z_i U,\;X_i^\text{S} := U^{-1} X_i U; \label{eq:ab}
    \end{equation}
    that is, $X_i$ and $X_i^\text{S}$, and $Z_i$ and $Z_i^\text{S}$, are unitarily similar, and in the language of group theory, this is conjugation by $U$ \cite{RomanG}. The group generated by $X_i$ and $Z_i$ is $I_\mathrm{C}\otimes \mathsf{P}^m$, which is a faithful representation of $\mathsf{P}^m$, so the group generated by $Z_i^\text{S}$'s and $X_i^\text{S}$'s, denoted by $\mathsf{P}^m_\text{S}$, is also a faithful representation of $\mathsf{P}^m$: $I_\mathrm{C}\otimes \mathsf{P}^m$ and $\mathsf{P}^m_\text{S}$ are unitarily equivalent representations, i.e.,
    \begin{equation}
        \mathsf{P}_{\mathrm{S}}^m := U^{-1} \left( I_{\mathrm{C}} \otimes \mathsf{P}^m \right) U.
    \end{equation}

    With these observations, the proposition is proved:
    \begin{enumerate}
        \item $I_{\mathrm{C}} \otimes \mathsf{P}^m$ is a group of Paulian operators, so is $\mathsf{P}_{\mathrm{S}}^m$---Note unless $\hils{C}$ is (isomorphic to) $\mathbb{C}^{2^p}$ for some integer $p$, $I_{\mathrm{C}} \otimes \mathsf{P}^m$ is not a group of Pauli operators. Besides, $Z_i$'s, $m$ in total, generate a maximal linearly independent and abelian subgroup\footnote{Please see the discussion near the end of Sec.~\ref{sec:phaseless}.} of $I_{\mathrm{C}} \otimes \mathsf{P}^m$; by unitary equivalence, $Z_i^\mathrm{S}$'s, $m$ in total, also generate a maximal linearly independent and abelian subgroup of $\mathsf{P}_{\mathrm{S}}^m$:
        \begin{equation}
            \mathsf{S} := \langle Z_1^\text{S},\dotsc,Z_m^\text{S} \rangle.
        \end{equation}
        
        \item Because $\hil_\mathrm{C} \otimes \ket{(t)}$ are the $(t)$-simultaneous eigenspaces of $Z_i$'s, $\hil_{(t)}$ are the $(t)$-simultaneous eigenspaces of $Z_i^\mathrm{S}$'s. $\hils{C} = \hil_{(I)}$ is hence stabilized by $Z_i^\mathrm{S}$'s.
        \item For any $\ket{\psi} \in \hils{C}$, if $F_{(t)} \in \mathbb{F}'$ occurs, $\ket{\psi} \in \hil_\mathrm{C}$ becomes $F_{(t)}\ket{\psi} \in \hil_{(t)}$ and it is a $(t)$-simultaneous eigenvector of $Z_i^\mathrm{S}$'s; performing the syndrome measurement by measuring $Z_i^\mathrm{S}$'s we obtain the simultaneous eigenvalues $(t)$, which are the \textbf{error syndrome} \cite{Gottesman10}, and we can correct the error by inverting $F_{(t)}$. Hence, any correctable error $E$ for which 
        \begin{equation}
            E\hil_\mathrm{C}\subseteq\overline{\hil}
        \end{equation}
        can be detected and corrected by measuring $Z_i^\mathrm{S}$.
    \end{enumerate}
\end{proof}

In a nutshell, via the isomorphism \eqref{eq:binh} and $U$ of \eqref{eq:U}, we borrow the structure from $\hil_\text{B}$ and apply it to $\hil' = \overline{\hil} \subseteq \hil$: $Z_i^\text{S}$ and $X_i^\text{S}$ are essentially Pauli $Z$ and $X$ on different subsystems or sites, and such a structure can be established for any quantum error-correcting codes. Treating $\overline{\hil}$ and $\hil_\text{B}$ as identical, $\mathbb{C}_i^2$'s of $\hil_\mathrm{B}$ are the stabilizer qubits \cite{Poulin05}. For Pauli stabilizer codes if we consider $\overline{\hil} = \hil$ and $\hil_\text{B}$ as the same space, the unitary map $U$, which becomes an operator now, is in the Clifford group \cite{Calderbank98,Gottesman98,Gottesman10,Poulin05}. 

We will call members of $\mathsf{S}$ \textbf{Paulian stabilizers}. Like Pauli stabilizer codes, we can choose any generating set of $\mathsf{S}$ for syndrome measurements. From now on, rather than \eqref{eq:ht}, $\hil_{(t)}$ will refer to the $(t)$-simultaneous eigenspace of $Z_i^\mathrm{S}$'s, and we will call it a $(t)$-\textbf{syndrome space}. Defining them this way will help us extend the Paulian stabilizers later.

\subsection{A larger stabilizer group} \label{sec:lastb}

The stabilizers depicted in Sec.~\ref{sec:stbgp} are the minimal version of Proposition~\ref{pro:main} with $\hil' = \overline{\hil}$, as the procedures laid out above are applicable to every code; however, when $\log_2 \left|\mathbb{F}\right|$ is not an integer, $\mathbb{F}' \subset \mathbb{F}$, and there are correctable errors that cannot be detected by $Z_i^\mathrm{S}$'s. 

Now suppose the code obeys
\begin{equation}
    2^{\left \lceil \log_2 \left|\mathbb{F}\right| \right \rceil} \dim \hils{C} \leq \dim \hil, \label{eq:m'+1}
\end{equation}
where $\left \lceil \cdot \right \rceil$ is the ceiling function. Let
\begin{equation}
    m = \left \lceil \log_2 \left|\mathbb{F}\right| \right \rceil,
\end{equation}
and we can consider a larger family of orthogonal operators $\mathbb{F}''$ such that $\mathbb{F} \subseteq \mathbb{F}''$, and that in addition to errors in $\mathbb{F}$ obeying \eqref{eq:errEi} we require 
\begin{equation}
    E \hils{C} \cong \hils{C} \text{ and } E\hils{C} \perp F\hils{C} \;\forall E\neq F \in \mathbb{F}''.
\end{equation}
Like before, for each element in $\mathbb{F}''$ we will associate with it a unique binary tuple of length $m$, i.e., a bijection between $\mathbb{F}''$ and $\mathbb{T}$; cf. Sec.~\ref{sec:stbgp}. This way, the Paulian stabilizers associated with $\mathbb{F}''$ covers all errors in $\mathbb{F}$. The operators in $\mathbb{F}''\setminus \mathbb{F}$ may be uncorrectable as they may not satisfy \eqref{eq:errEi}, but they are instrumental in constructing a larger Paulian stabilizer group. 

In short, we would like the Paulian stabilizers to cover all correctable errors, hence choosing $m = \left \lceil \log_2 \left|\mathbb{F}\right| \right \rceil$ if possible; if \eqref{eq:m'+1} cannot be satisfied, we resort to $m = \left \lfloor \log_2 \left|\mathbb{F}\right| \right \rfloor$. In particular, given a code with distance $d$, we have
\begin{equation}
    \left|\mathbb{F}\right| \leq \sum_{j = 0}^{\left\lfloor (d-1)/2 \right\rfloor} \begin{pmatrix}
        n \\
        j
    \end{pmatrix}
    3^j, \label{eq:fcos}
\end{equation}
which serves as an upper bound for $\left|\mathbb{F}\right|$ and is exact when the code is nondegenerate \cite{Gottesman97,Gottesman10}, cf. the quantum Hamming bound \cite{Gottesman10,Preskill}. Expression \eqref{eq:fcos} combined with \eqref{eq:m'+1}  is a sufficient condition to judge whether it is possible to find Paulian stabilizers to correct all the errors for this code, explicitly,
\begin{equation}
    2^{\left \lceil \log_2 \sum_{j = 0}^{\left\lfloor (d-1)/2 \right\rfloor} \begin{pmatrix}
        n \\
        j
    \end{pmatrix}
    3^j \right \rceil} \dim \hils{C} \leq \dim \hil, \label{eq:fcos2}
\end{equation}
which is necessary and sufficient if the code is nondegenerate.
\subsection{Extending the domain}
\label{sec:extdo}
If $2^m \dim \hils{C} < \dim \hil$, we may extend the domains of $Z_i^\mathrm{S}$'s, and they should remain self-adjoint so that they are measurable. The following fact may be utilized \cite{Axler}:
\begin{theorem}\label{thm:seadop}
    For a self-adjoint/unitary operator $A$ with an invariant subspace $\hil'$, $A|_{\hil'}$ and $A|_{{\hil'}^\perp}$ are both self-adjoint/unitary operators.
\end{theorem}
\noindent Hence, to extend a self-adjoint operator, we can define another self-adjoint operator on the space orthogonal to the domain and add them.

In particular, let's extend $Z_i^\mathrm{S}$'s as follows: We enlarge all syndrome spaces $\hil_{(t)}$'s while keeping them isomorphic and orthogonal to each other, and let their dimension be $k'$, which would be no smaller than $\dim\hils{C}$. Following how $Z_i^\mathrm{S}$'s were originally constructed on $\overline{\hil}$ in Sec.~\ref{sec:stbgp}, we reach the final form of Proposition~\ref{pro:main}, where $\hil'$ is the direct sum of all syndrome spaces and $Z_i^\mathrm{S}$'s are commutative Paulian operators to the restriction of $\hil'$. Thus, the proof in Sec.~\ref{sec:stbgp} and the discussion in Sec.~\ref{sec:lastb} and this subsection together illustrate the complete picture of Proposition~\ref{pro:main}. 

We remark
\begin{enumerate}
    \item $F_{(t)}\hils{C}$ is a subspace of the corresponding syndrome space $\hil_{(t)}$, in particular
    \begin{equation}
        \hils{C} \subseteq \hil_{(I)}:
    \end{equation}
    The code space is stabilized by $Z_i^\mathrm{S}$'s, but it is not necessarily the $(I)$-syndrome space, but a subspace thereof. 
    \item If $\hil$ is infinite-dimensional, syndrome spaces can be made infinite-dimensional while keeping them isomorphic and mutually orthogonal; an example will be given in Sec.~\ref{sec:boco}.
    \item When $\dim \hil/\dim \hils{C}$ is an integral power of 2, it is always possible to construct a Paulian stabilizer group that uses the space to its full capacity for error correction; specifically, this is true for binary codes. 
\end{enumerate}

\subsection{Measuring Paulian operators}
\label{sec:cnot}
Suppose the state is currently in $\hil'$. To measure a $Z_i^\mathrm{S}$, we can make use of a generalized CNOT: Consider an ancilla qubit $\hils{A} \cong \mathbb{C}^2$ initialized at $\ket{1}_\mathrm{A}$. To the restriction of $\hil' \otimes \hils{A}$, define
\begin{align}
    \mathrm{GCNOT}&|_{\hil' \otimes \hils{A}} := \Pi_{-} \otimes X_\mathrm{A} + \Pi_{+} \otimes I_\mathrm{A} \label{eq:cnotano} \\
    &= I_{\hil'} \otimes \ket{+}_\mathrm{A}\bra{+}_\mathrm{A} + Z_i^\mathrm{S} \otimes \ket{-}_\mathrm{A}\bra{-}_\mathrm{A}, \label{eq:cnot}
\end{align}
where $\Pi_{\pm}$ project onto the $\pm 1$-eigenspaces of $Z_i^\mathrm{S}$ and $\ket{\pm}_\mathrm{A}$ are $\pm 1$-eigenstates of $X_\mathrm{A}$; in Appendix~\ref{sec:gcnotex} it will be explained why \eqref{eq:cnotano} and \eqref{eq:cnot} are equal and how unique the generalized CNOT is. On the entire space $\hil \otimes \hils{A}$ it is thus
\begin{equation}
    \mathrm{GCNOT} = \Pi_{-} \otimes X_\mathrm{A} + \Pi_{+} \otimes I_\mathrm{A} + \Pi_{{\hil'}^\perp}\otimes U_\mathrm{A}, \label{eq:cnot1}
\end{equation}
where $U_\mathrm{A}$ can be any unitary operator; $\Pi_{{\hil'}^\perp}\otimes U_\mathrm{A}$ is there for GCNOT to be unitary, cf. Theorem~\ref{thm:seadop}. If the system is in a $-1$-eigenstate of $Z_i^\mathrm{S}$, the state of the qubit will be mapped to $\ket{-1}_\mathrm{A}$, else it remains at $\ket{1}_\mathrm{A}$, so measuring $Z_\mathrm{A}$ on the ancilla afterwards is equivalent to measuring $Z_i^\mathrm{S}$. This operator derives from the generalized CNOT or controlled-$X$ in Refs.~\cite{Di13,Di15,Pavlidis21,Saha22}, but we do not require that the system and ancilla have the same dimension. From now on for simplicity we will ignore the restriction.

Quite many implementations of Pauli measurements involve these controlled operations implicitly \cite{Dennis02,Fowler12,Terhal15}: For example, to measure $Z_1\cdots Z_j$, with the (regular) CNOT on the $i$-th (data) qubit with the ancilla as the target denoted by $\text{CNOT}_i$, it can be found
\begin{equation}
    \bigotimes_{i=1}^{j} \text{CNOT}_i = \Pi_-\otimes X_\mathrm{A} + \Pi_+\otimes I_\mathrm{A}: \label{eq:cnotcom}
\end{equation}
$\Pi_{\pm}$ are the $\pm 1$-eigenspaces of $Z_1\cdots Z_j$, so the composition is a generalized CNOT.

Using a non-qubit system as the control may not be as intuitive, so let's instead consider the controlled-$Z_i^\mathrm{S}$:
\begin{equation}
    CZ_i^\mathrm{S} = Z_i^{\mathrm{S}} \otimes \ket{-1}_\mathrm{A}\bra{-1}_\mathrm{A} + I_{\hil'} \otimes \ket{1}_\mathrm{A}\bra{1}_\mathrm{A}, \label{eq:CZ}
\end{equation}
which is a controlled-$U$ operation with $U$ being the Paulian operator $Z_i^\mathrm{S}$ \cite{Nielsen}. Compared with \eqref{eq:cnot}, we have
\begin{equation}
    \mathrm{GCNOT} = \left( I_{\mathrm{H}'} \otimes H_\mathrm{A} \right) CZ_i^\mathrm{S} \left( I_{\mathrm{H}'} \otimes H_\mathrm{A} \right), \label{eq:cnotcz}
\end{equation}
where $H_\mathrm{A}$ is the Hadamard gate. In other words, if the ancilla is initialized at $\ket{1}_\mathrm{A}$, we can perform an inverse Hadamard gate to map it to $\ket{+}_\mathrm{A}$, and apply the $CZ_i^\mathrm{S}$ gate. Measuring $X_{\mathrm{A}}$ on the ancilla, if the result is $\pm 1$, then it means the system was in a $\pm 1$-eigenstate of $Z_i^\mathrm{S}$, so the overall effect is identical to measuring $Z_i^\mathrm{S}$. 

Equation \eqref{eq:cnotcz} is in the same vein as exchanging the target and control qubits of a (regular) CNOT by composing with Hadamard gates or change of basis \cite{Nielsen}: Indeed, \eqref{eq:cnot} can be understood as using the $\ket{\pm}_\mathrm{A}$ states of the ancilla qubit to determine whether to perform $Z_i^\mathrm{S}$, so the ancilla qubit is the control in this sense; with \eqref{eq:cnotcz} we simply change the ``control states'' from $\ket{\pm}_\mathrm{A}$ to $\ket{\pm 1}_\mathrm{A}$.

Regarding the ancilla qubit as the control as in \eqref{eq:cnot} or \eqref{eq:CZ} also brings the following benefit: Suppose the system is composed of qubits, and that we have a quantum circuit for $Z_i^\mathrm{S}$ using fundamental gates, comprising single-qubit gates and CNOT or two-qubits controlled gates; let $Z_i^\mathrm{S} = \prod_i U_i$, and we have 
\begin{equation}
    CZ_i^\mathrm{S} = \prod_j CU_j.
\end{equation}
If $U_i$ is a single-qubit gate, then we can again decompose $CU_j$ as single-qubit gates and CNOT's; if $U_j$ is a CNOT or a controlled-$V_j$ for some $V_j$, then $CU_j$ is a Toffoli gate or $C^2(V_j)$ and we can again decompose it as fundamental gates \cite{Nielsen,Barenco95,Liu08}. Thus, if we are able to carry out $Z_i^\mathrm{S}$ as an operation then we are also able to measure it. Equation \eqref{eq:cnotcom} can also be better comprehended with the ancilla as the control.

We discussed extending Paulian-ness to a larger space $\hil'$ in Sec.~\ref{sec:extdo}. From a measurement point of view, how the Paulian stabilizers should be extended depends on whether the corresponding controlled operations are natural, that is whether we can couple the system with the ancilla via the controlled operations relatively easily.

We can tweak \eqref{eq:cnot} to use an ancilla qudit (which may be composed of several qubits) for the setup to be less error-prone \cite{Gottesman97,Gottesman98,Gottesman10}: Omitting $\Pi_{{\hil'}^\perp}\otimes U_\mathrm{A}$ again, let the system be coupled with the qudit initialized at $\ket{1}_\mathrm{A}$ through this generalized CNOT
\begin{equation}
    \sum_{i_+,i_-} \left( \ket{i_-}\bra{i_-} \otimes X_{i_-} + \ket{i_+}\bra{i_+} \otimes X_{i_+} \right), \label{eq:cnot2}
\end{equation}
where 
\begin{enumerate}
    \item $\left\{ \ket{i_\pm} \right\}$ are orthonormal bases of the $\pm 1$-eigenspaces of $Z_i^\mathrm{S}$,
    \item each $X_{i_\pm}$ is an X operator between the states $\ket{1}_\mathrm{A}$ and $\ket{i_\pm}_\mathrm{A}$ of the qudit \cite{Di13}, and
    \item we have an observable $Z'_\mathrm{A}$ on the qudit, with $\ket{i_\pm}_\mathrm{A}$ being $\pm 1$-eigenstates of $Z'_\mathrm{A}$; $\ket{i_-}_\mathrm{A}$'s may not be orthogonal with or even different from each other, likewise for $\ket{i_+}_\mathrm{A}$'s. $\ket{1}_\mathrm{A}$ is a $+1$-eigenstate of $Z'_\mathrm{A}$, and some of the $\ket{i_+}_\mathrm{A}$'s may be $\ket{1}_\mathrm{A}$.
\end{enumerate}
Afterwards, we measure $Z'_\mathrm{A}$, and all this combined is equivalent to measuring the Paulian operator. Expression \eqref{eq:cnot2} is again an adaption of the generalized CNOT from Refs.~\cite{Di13,Di15,Pavlidis21,Saha22}.

\subsection{Subsystem codes} 
Subsystem codes can be considered a generalization of regular error-correcting codes \cite{Kribs05,Poulin05,Kribs06,*Kribs06arXiv,Bacon06,Aly06,Aliferis07,Higgott21}: The code space becomes
\begin{equation}
    \hils{C} = \hils{L}\otimes \hils{G} \subseteq \hil.
\end{equation}
The information is stored in the logical subsystem $\hils{L}$ and the state of the gauge subsystem $\hils{G}$ does not matter. 

Proposition~\ref{pro:main} is also applicable to subsystem codes, explicitly:
\begin{corollary}\label{cor:sub}
    The code space of every subsystem code can be stabilized by operators with properties identical to those listed in Proposition~\ref{pro:main}. Hence, after obtaining the Paulian stabilizer group $\mathsf{S}$, for any nonzero $A\in \lmap{\hils{G}}$, all elements in $(I_\mathrm{L}\otimes A)\mathsf{S}$ leave the encoded state intact.
\end{corollary}

Let's provide a simple argument as to why this is true: According to Ref.~\cite{Kribs06PRA}, with $\mathbb{E}$ denoting the set of correctable errors, for a subsystem code it is possible to find a set $\mathbb{F}$ of ``orthogonal'' correctable errors that obeys \eqref{eq:ehc}, just like an ordinary error-correcting code. Utilizing $\mathbb{F}$ and the corresponding syndrome spaces, we can obtain Paulian stabilizers by following the steps in Sec.~\ref{sec:stbgp}.

\section{Concatenation of binary codes}
\label{sec:con}

Binary codes with appropriate parameters can be concatenated, and we will illustrate how to acquire the Paulian stabilizer group of the new code, in a similar fashion to Pauli stabilizer codes \cite{Preskill}. A symbol with sub- or superscript in, out, or $+$ ($+$ for ``adding'') indicates it belongs to the inner, outer, or concatenated codes respectively, and the sub- or superscript $w$ means it can be one of those three.

Let them be $[[n_w,k_w]]$-codes, and the inner and outer codes have Paulian stabilizers $Z_{i}^w$'s. To concatenate them,
\begin{equation}
    q := n_\text{out}/k_{\text{in}}
\end{equation}
should be an integer. Let $\hil_w$ be the space each code belongs in, and
\begin{equation}
    \hil_+ = \hil_\text{in}^{\otimes q}. \label{eq:con}
\end{equation}
Define the following operators on $\hil_+$:
\begin{equation}
    Z_{i,j}^+ := \underbrace{I_\text{in}\otimes \cdots \otimes I_\text{in}}_{j-1 \text{ subsystems}}\otimes \underbrace{Z^{\text{in}}_i}_{j\text{-th }\hil_\text{in}} \otimes I_\text{in}\otimes\cdots\otimes I_\text{in}, \label{eq:aij}
\end{equation}
which are independent and commute with each other, and which stabilize
\begin{equation}
    \hil_\text{in,C}^{\otimes q} \subseteq \hil_+,
\end{equation} 
where $\hil_\text{in,C}$ is the code space of the inner code. 

Next, let the logical Pauli operators on the inner code commute with all $Z_{i}^\text{in}$'s, and expand every $Z_{i}^\text{out}$ in terms of Pauli operators. Since $\hil_\text{in,C}$ and $\hil_\text{out}$ are composed of $k_\text{in}$ and $q k_\text{in}$ qubits respectively, regarding every $k_\text{in}$ qubits of $\hil_\text{out}$ as $\hil_\text{in,C}$ we can replace each Pauli operator in every $Z_{i}^\text{out}$ by the corresponding logical Pauli operator on the inner code, and the resultant operator will be denoted by $Z_{i}^+$. $Z_{i}^+$'s together with $Z_{j,k}^+$'s are mutually commutative and independent Paulian operators, composing the stabilizer generators of the concatenated code. 

How to get $Z_{i}^+$'s may be a little hard to comprehend, so here's a quick demonstration: Suppose $n_\mathrm{out} = 4$ and $k_\mathrm{in} =2 $. As $q=4/2=2$,
\begin{equation}
    \hil_+ = \hil_\text{in}\otimes \hil_\text{in}. \label{eq:hcon}
\end{equation}
If $Z_{1}^\text{out} = (X \otimes X \otimes Z \otimes Y + Z \otimes Z \otimes X \otimes I_2)/2$, then
\begin{equation}
    Z_{1}^+ = \left( \overline{X}_1 \overline{X}_2 \otimes \overline{Z}_1 \overline{Y}_2 + \overline{Z}_1 \overline{Z}_2 \otimes \overline{X}_1 I_\text{in} \right)/2, \label{eq:a1con}
\end{equation}
where $\overline{L}_i$ denotes the logical $L$ operator of the $i$-th logical qubit. In \eqref{eq:a1con}, $\overline{X}_1 \overline{X}_2$ and $\overline{Z}_1 \overline{Z}_2$ act on the first $\hil_\text{in}$ of \eqref{eq:hcon}, and $\overline{Z}_1 \overline{Y}_2$ and $\overline{X}_1 I_\text{in}$ on the second one. Notice logical operators acting on different logical qubits commute with one another, so e.g. $\overline{X}_1 \overline{X}_2 = \overline{X}_2 \overline{X}_1$. Also a logical identity operator is simply the identity operator on the system, so in \eqref{eq:a1con} instead of $\overline{X}_1 \overline{I}_2$ we had $\overline{X}_1 I_\text{in} = \overline{X}_1$; we spelled $I_\text{in}$ out for clarity.

The methods to find the parameters and codewords of concatenated codes are well established \cite{Knill96,Gottesman97,Grassl09,Wang13}, on which we will provide a short discussion in Appendix~\ref{sec:cococo}.

\section{Examples}
\label{sec:ex}
Here we will show the Paulian stabilizers of some codes, or how to find them. 

\subsection{Transformed from Pauli stabilizer codes}
Given a Pauli stabilizer code that can correct a set of operators $\mathbb{E}$ with stabilizers $Z_i^\mathrm{S}$, we can perform a unitary transformation $U$ on the system, which can be seen as a change of orthonormal basis. The transformed stabilizer generators
\begin{equation}
    {Z_i^\mathrm{S}}' = U Z_i^\mathrm{S} U^{-1}
\end{equation}
will be Paulian, and they can correct a set of operators $U \mathbb{E} U^{-1}$. If the unitary transformation is local in each (physical) qubit, the distance shall stay the same. 

For illustration, consider an $n$-qubit repetition code with stabilizers \cite{Preskill,Gunther22}
\begin{equation}
    Z_1 Z_2,\dotsc,Z_{n-1}Z_n ,
\end{equation}
which can fix an $X$ error on every qubit. We can construct a generalized repetition code for normal operators with
\begin{lemma}\label{lem:i+u}
    An operator on $\mathbb{C}^2$ is normal if and only if it can be a linear combination of the identity and a Paulian operator that is not proportional to the identity. Note as the Paulian operator is not proportional to the identity, it has two eigenvalues.
\end{lemma}
\noindent The proof can be found at Appendix~\ref{sec:prlem}. From the discussion in Secs.~\ref{sec:inv} and \ref{sec:pau}, the Paulian operator in Lemma~\ref{lem:i+u} can be chosen to be self-adjoint so that it has eigenvalues $1$ and $-1$.

By this lemma, consider any normal operator $E = a I + b V$ on $\mathbb{C}^2$, where $a,b\in \mathbb{C}$ and $V$ is self-adjoint and Paulian. After obtaining $V$, as $X$ and $V$ have the same spectrum, $X$ and $V$ are unitarily similar via some unitary $U$; in other words,  we can perform local unitary transformations $V_i$ such that $V_i = U_i X_i U_i^{-1}$ for all the (physical) qubits, where the subscript $i$ indicates which qubit the operator act on nontrivially. This way, we acquire Paulian stabilizers that can correct $E$ on a single qubit: 
\begin{equation}
    \left( U_i Z_i U_i^{-1} \right) \left( U_{i+1} Z_{i+1} U_{i+1}^{-1} \right), \; i = 1,\dotsc,n-1.
\end{equation}
Hence, we can correct any error that is a normal operator on a single qubit; specifically, the normal operator can be any unitary operator.

\subsection{Bosonic codes} \label{sec:boco}

Let's first consider this bosonic binomial code \cite{Michael16,Hu19,Terhal20}:
\begin{equation}
    \ket{\overline{1}} := \ket{2},\; \ket{\overline{-1}} := \left( \ket{4} + \ket{0} \right)/\sqrt{2},
\end{equation}
which was experimentally demonstrated in Ref.~\cite{Hu19} and can correct orthogonal errors $I$ and the annihilation operator $a$. The space can be partitioned according to parity \cite{Sun14,Michael16,Hu19,Terhal20}---The code space $\hils{C}\subset \hil_{(1)}$ has even parity while $a\hils{C}\subset \hil_{(-1)}$ has odd parity. With $N := a^\dagger a$, the parity operator
\begin{equation}
    Z^\mathrm{S} = e^{i\pi N} \label{eq:parity}
\end{equation}
is actually Paulian (see Sec.~\ref{sec:pau}): 
\begin{enumerate}
    \item It is clear that \eqref{eq:parity} is unitary;
    \item Because its eigenvalues are $\pm 1$, it is an involution;
    \item Finally, because $\left\{\ket{2n-2}\right\}_{n\in \mathbb{N}}$ and $\left\{\ket{2n-1}\right\}_{n\in \mathbb{N}}$ are the bases of its $\pm 1$-eigenspaces, their orthonormal bases have the same cardinality---Namely we can establish a bijection between $\left\{\ket{2n-2}\right\}_{n\in \mathbb{N}}$ and $\left\{\ket{2n-1}\right\}_{n\in \mathbb{N}}$---The two eigenspaces are thus isomorphic.
\end{enumerate}

The parity operator is Paulian on the whole space, i.e., $\hil' = \hil$, which is infinite-dimensional. To measure the syndrome, the controlled phase gate $I\otimes \ket{-1}\bra{-1} + e^{i\pi N} \otimes \ket{1}\bra{1}$ \cite{Vlastakis13,Sun14} is the controlled operations \eqref{eq:cnot} and \eqref{eq:CZ} along with appropriate rotation on the ancilla.

The next one is the bosonic code from Ref.~\cite{Chuang97}:
\begin{equation}
    \ket{\overline{1}} := \ket{22},\; \ket{\overline{-1}} := \left( \ket{40} + \ket{04} \right)/\sqrt{2},
\end{equation}
which protects up to one photon loss, and we have the following orthogonal correctable errors \cite{Chuang97,Terhal20}:
\begin{equation}
    \mathbb{E} = \left\{ I,\,A_{1,1},A_{1,2} \right\},
\end{equation}
where $A_{i,j}$ is the damping operator for which the $j$-th mode losing $i$ photons (hence $I$ corresponds to $A_{0,1}$ and $A_{0,2}$). We can again choose parity operators as the Paulian stabilizers:
\begin{equation}
    Z_1^\mathrm{S} = e^{i\pi N_1}, \; Z_2^\mathrm{S} = e^{i\pi N_2}.
\end{equation}
The correctable errors $I$, $A_{1,1}$ and $A_{1,2}$ will have syndromes $(I) = (1,1)$, $(-1,1)$ and $(1,-1)$ respectively. Note in this case, we not only extend the domains to the entire space but also enlarge the Paulian group as $2 = \lceil \log_2 3 \rceil > \lfloor\log_2 3 \rfloor =1$; see Secs.~\ref{sec:lastb} and \ref{sec:extdo}. Photon loss for the bosonic four-legged cat code can also be detected by parity \cite{Leghtas13,Mirrahimi14,Ofek16,Terhal20}, so we also have a Paulian stabilizer for such a code. 

Finally, in Appendix~\ref{sec:gkp} we will have a brief discussion about Gottesman-Kitaev-Preskill codes \cite{Gottesman01,Campagne20,Terhal20}, where we will show a way to construct commutative Paulian stabilizers for these codes and issues with them.

\subsection{Codeword stabilized code} \label{sec:csc}
For an $n$-qubits system, a codeword stabilized code \cite{Cross09,Chuang09,Chen08PRA} is obtained in the following way:
\begin{enumerate}
    \item We start with a maximally linearly independent and abelian subgroup of a Pauli group (please refer to Appendix~\ref{sec:commpa} for the exact meaning), called the \textbf{word stabilizer}. 
    \item We also need a set of Pauli operators $\left\{ W_i \right\}$, called the \textbf{word operators}.
    \item As the word stabilizer is maximally linearly independent and abelian, each of its simultaneous eigenspace is one-dimensional, i.e., it stabilizes a unique quantum state; let it be $\ket{\psi}$.
    \item The codewords are then $W_i\ket{\psi}$'s; that is, the code space is $\mathrm{span} \left\{ W_i\ket{\psi} \right\}$.
\end{enumerate}
The following result can be utilized to construct Paulian stabilizers of a codeword stabilized code:
\begin{corollary}\label{cor:csc}
    For a codeword stabilized code:
    \begin{enumerate}
        \item If $P_1$ and $P_2$ are correctable Pauli errors, they are either orthonormal or act identically on the code space bar a multiplication factor.
        \item Assuming the code has distance $d$, it is nondegenerate \cite{Gottesman97,Gottesman10} if and only if every operator in the word stabilizer except $I$ has distance no smaller than $d$.
    \end{enumerate}
\end{corollary}

In Appendix~\ref{sec:costb}, specifically Sec.~\ref{sec:conpastb}, we provide a procedure to construct Paulian stabilizers that is applicable to every codeword stabilized code; here is the essence: We first determine whether there exist Paulian stabilizers that can correct all the relevant errors by \eqref{eq:m'+1} or \eqref{eq:fcos2}; then according to Corollary~\ref{cor:csc} we can choose linearly independent Pauli errors as orthonormal correctable errors, and to be definite we can check whether the code is nondegenerate again by Corollary~\ref{cor:csc}. We then use the simultaneous eigenspaces of the word stabilizer to build syndrome spaces, which lead to Paulian stabilizers.

An example would be the $((9,12,3))$-code from Refs.~\cite{Yu08,Cross09}: Each element of the word stabilizer except $I$ has at least weight $3$, so the code is nondegenerate according to Corollary~\ref{cor:csc}, and we can choose all linearly independent weight-1 Pauli errors, along with $I$ as the orthonormal errors $\mathbb{F}$ of \eqref{eq:errEi}. By \eqref{eq:fcos}, 
$$\left| \mathbb{F} \right| = 3\times 9 + 1 = 28,$$
so \eqref{eq:m'+1} or [\eqref{eq:fcos2}] is satisfied, and we can construct a Paulian stabilizer group to correct all the relevant errors, generated by 
$\log_2 \left| \mathbb{F} \right| = 5$
Paulian operators. Furthermore, we can extend the stabilizers so that each is Paulian on the whole space. 

In fact, ``Paulian stabilizers'' for this code have already been found in Ref.~\cite{Yu08}, among which some are Pauli.\footnote{That this code can be stabilized by nontrivial Pauli operators can also be verified with Corollary~\ref{cor:stbco} in the Appendix.} Note, however, that the ``Paulian stabilizers'' from Ref.~\cite{Yu08} possess a different structure from those presented in this work: The Paulian stabilizers of Proposition~\ref{pro:main} are elements of a faithful representation of the Pauli group, so they are reminiscent of Pauli stabilizers of a Pauli stabilizer code. On the other hand, those from Ref.~\cite{Yu08} are not, so different sequences of measurements are needed for different errors, and more than five observables are needed to detect all the errors, whereas with the Paulian stabilizers of Proposition~\ref{pro:main} we require only five commutative observables for measurement. Even though Paulian stabilizers like those in Ref.~\cite{Yu08} are interesting and useful \emph{per se}, we will not delve into them. More details about this code can be found in Appendix~\ref{sec:9123}. 

Now let's consider the $((5,6,2))$-code from Refs.~\cite{Rains97,Cross09}. Due to its distance, this code is an error-detecting code. It can be found that, with $\mathbb{P}_1$ denoting the set of all weight-1 Pauli errors, we have
\begin{equation}
    \hils{C}^\perp = \sum_{P\in \mathbb{P}_1} P\hils{C},
\end{equation}
which implies for the stabilizers to detect all errors in $\mathbb{P}_1$, we must have
\begin{align}
    \hils{C} &= \hil_{(I)} \label{eq:i(I)} \\
    \hils{C}^\perp &= \bigoplus_{(t)\in \mathbb{T}\setminus \{(I)\}} \hil_{(t)}, \label{eq:ct}
\end{align}
where $\mathbb{T}$ is the set of all syndromes; see Sec.~\ref{sec:stbgp}. For the stabilizers to be Paulian and commutative, each syndrome space must have the same dimension, so \eqref{eq:i(I)} and \eqref{eq:ct} together imply
\begin{equation}
    \dim \hil = 2^m \dim\hils{C}
\end{equation}
for some positive integer $m$, which is impossible for this system as $\dim\hils{C} = 6$ and $\dim \hil = 2^5 = 32$. 

Hence, we cannot find commutative Paulian stabilizers for the $((5,6,2))$-code to detect all weight-1 errors---This is one of the cases where Paulian stabilizer groups may not be suitable for error correction or detection, cf. the discussion in Sec.~\ref{sec:lastb}. Regardless, because this code has low dimensions, it is easier to demonstrate how to find its Paulian stabilizers as every step can be made explicit without being too clumsy; in addition, we can show how to adapt our approach to error-detecting codes. Details can be found in Appendix~\ref{sec:562}.

\section{Discussion and conclusion} \label{sec:dis}

We showed that every quantum error-correcting code, including the subsystem code, can be stabilized by operators which are Paulian and commutative to the restriction of a subspace $\hil'$, which may or may not be the entire system $\hil$ (Proposition~\ref{pro:main} and Corollary~\ref{cor:sub}), with examples given in Sec.~\ref{sec:ex}. In addition, we showed that the error syndrome can be obtained by measuring the Paulian stabilizers $Z_i^\mathrm{S}$'s, which can be achieved by performing controlled operations $CZ_i^\mathrm{S}$'s, so the quantum circuits for conducting $Z_i^\mathrm{S}$'s can be transferred to those for measuring them (Sec.~\ref{sec:cnot}).

In terms of tensor product structure \cite{Zanardi97,Zanardi01PRL,Zanardi04PRL}, $\hil'$ is composed of $m$ stabilizer qubits \cite{Poulin05} generated by the Paulian operators, and a subsystem isomorphic to the syndrome spaces, whose dimension $k'$ is no less than $\dim \hils{C}$, so we can embed $\hils{C}$ into them. This generalizes the observation made in Refs.~\cite{Zanardi04PRL,Poulin05}, that for a system composed of qubits, commutative Paulian operators can partition the system into virtual qubits; if the Paulian operators are Pauli it becomes a Pauli stabilizer code. 

Paulian stabilizers may be employed to realize codes that are not Pauli stabilizer codes, showcased in Sec.~\ref{sec:boco}. As discussed in Sec.~\ref{sec:lastb}, \eqref{eq:m'+1} is the condition for Paulian stabilizers to cover all correctable errors. Hence, binary codes may in particular benefit from the existence of Paulian stabilizers, because \eqref{eq:m'+1} is always satisfied; the same is true in the case where the code space is finite-dimensional while the entire system is infinite-dimensional, such as the bosonic codes in Sec.~\ref{sec:boco}. Furthermore, as we have demonstrated how to obtain the Paulian stabilizer group of a binary concatenated code in Sec.~\ref{sec:con}, it may help us obtain a code with higher distance along with the means to realize it.

There are questions still left unanswered that may be worthy of further investigation: There is no unique Paulian stabilizer group for a code, and the ideal Paulian stabilizers are those that are easy to measure or conduct---With a universal set of quantum gates, we can in theory approach them \cite{Nielsen,Preskill}, but it may need many gates to implement. Hence, for Paulian stabilizers to be useful, how to find the ideal ones is a key issue, which depends on the physical system in question. Also, we showed the existence of Paulian stabilizers for error-correcting codes, but knowing this, can it help us find nontrivial new codes by using Paulian operators that are not Pauli as the stabilizers or correctable errors?

\section*{Acknowledgments}
J.-Y.~K. would like to thank Prof. Chung-Hsien Chou for introducing him to this area of research, and we thank Dr.~Tanmay Singal for very fruitful discussion about various aspects of quantum error correction.
H.-S.G. acknowledges support from the National Science and
Technology Council, Taiwan under Grants 
No.~NSTC 112-2119-M-002-014, 
No.~NSTC 111-2119-M-002-006-MY3, 
No.~NSTC 111-2119-M-002-007,   
No.~NSTC 110-2627-M-002-002, 
No.~NSTC 111-2627-M-002-001,
and No.~NSTC 111-2627-M-002-006,
from the US Air Force Office of Scientific Research under
Award Number FA9550-23-S-0001,
and from the National Taiwan University under
Grant No.~NTU-CC-112L893404. 
H.-S.G. is also grateful for the support from the
“Center for Advanced Computing and Imaging in Biomedicine (NTU-112L900702)” through The Featured Areas Research Center Program within the framework of the Higher Education Sprout Project by the Ministry of Education (MOE), Taiwan, and
the support from the Physics Division, National
Center for Theoretical Sciences, Taiwan.

\begin{appendix}

\section{Generalized CNOT} \label{sec:gcnotex}

Consider any self-adjoint Paulian operator $P$ on $\hil$, and an ancilla qubit $\hils{A} \cong \mathbb{C}^2$. We define the corresponding generalized CNOT as the following unitary operator on $\hil \otimes \hils{A}$:
\begin{equation}
    \mathrm{GCNOT} := \Pi_{-} \otimes X_\mathrm{A} + \Pi_{+} \otimes I_\mathrm{A}, \label{eq:gcnot}
\end{equation}
where $\Pi_{\pm}$ are the orthogonal projections onto the $\pm 1$-eigenspaces of $P$ and A refers to the ancilla qubit. Here let's have a quick discussion about why GCNOT is equal to
\begin{equation}
    I_{\hil} \otimes \Pi_{+}^\mathrm{A} + P \otimes \Pi_{-}^\mathrm{A} ,\label{eq:gcnot2}
\end{equation}
where 
\begin{equation}
    \Pi_{\pm}^\mathrm{A} := \ket{\pm}_\mathrm{A}\bra{\pm}_\mathrm{A}
\end{equation}
are orthogonal projections onto the $\pm 1$-eigenspaces of $X_\mathrm{A}$.

If $P$ has only one eigenvalue, then $P$ is either $I_{\hil}$ or $-I_{\hil}$. For the former, it is fairly easy to see both \eqref{eq:gcnot} and \eqref{eq:gcnot2} are $I_\mathrm{\hil} \otimes I_\mathrm{A}$, so they are identical. For the latter, \eqref{eq:gcnot} becomes $I_\mathrm{\hil} \otimes X_\mathrm{A}$, whereas \eqref{eq:gcnot2} becomes 
\begin{align}
    I_{\hil} \otimes \Pi_{+}^\mathrm{A} - I_{\hil} \otimes \Pi_{-}^\mathrm{A} &= I_{\hil} \otimes \left( \Pi_{+}^\mathrm{A} - \Pi_{-}^\mathrm{A} \right) \nonumber \\
    &= I_{\hil} \otimes X_\mathrm{A},
\end{align}
so they are again the same.

If $P$ has two eigenvalues, namely $\pm 1$, then we have
\begin{align}
    \Pi_{-} \otimes X_\mathrm{A} + \Pi_{+} \otimes I_\mathrm{A} = & \Pi_- \otimes \left(\Pi_{+}^\mathrm{A} - \Pi_{-}^\mathrm{A} \right) \nonumber \\
    & + \Pi_+ \otimes \left(\Pi_{+}^\mathrm{A} + \Pi_{-}^\mathrm{A} \right) \nonumber \\
    = & \left( \Pi_+ + \Pi_- \right) \otimes \Pi_{+}^\mathrm{A} \nonumber \\
    &+ \left( \Pi_+ - \Pi_- \right)\otimes \Pi_{-}^\mathrm{A} \nonumber \\
    =& I_\hil \otimes \Pi_{+}^\mathrm{A} + P \otimes \Pi_{-}^\mathrm{A},
\end{align}
which is \eqref{eq:gcnot2}. 

In fact, as the relations above do not depend on the dimension of the eigenspaces of $X_\mathrm{A}$, we can replace the ancilla qubit with a system with even dimension and $X_\mathrm{A}$ with another Paulian operator, and identical results will hold.

Second, let's try to answer this question: Given a Paulian operator $P$, is the generalized CNOT of \eqref{eq:gcnot} the only [besides a global phase factor or a phase difference between the two terms on the right hand side of \eqref{eq:gcnot}] unitary operator that can achieve what we want of it? Specifically, let $\overline{\mathrm{GCNOT}}$ denote the ``most general'' GCNOT for $P$; the property we desire is
\begin{align}
    \left( I\otimes \Pi_{+}^\mathrm{A} \right)\overline{\mathrm{GCNOT}}\left( \ket{\psi}\otimes \ket{1} \right) &= e^{i\theta_+}\left( \Pi_{+}\ket{\psi} \right)\otimes \ket{1}, \nonumber \\
    \left( I\otimes \Pi_{-}^\mathrm{A} \right)\overline{\mathrm{GCNOT}}\left( \ket{\psi}\otimes \ket{1} \right) &= e^{i\theta_-}\left( \Pi_{-}\ket{\psi} \right)\otimes \ket{-1},\label{eq:gcnotge}
\end{align}
for every $\ket{\psi} \in \hil$, where $\theta_\pm \in [0,2\pi)$. Thus, 
\begin{align}
    \overline{\mathrm{GCNOT}}\left( \ket{\psi} \otimes \ket{1} \right) =& \left( I\otimes \Pi_{+}^\mathrm{A} + I\otimes \Pi_{-}^\mathrm{A} \right) \overline{\mathrm{GCNOT}} \nonumber \\
    & \left( \ket{\psi} \otimes \ket{1} \right) \nonumber \\
    = & e^{i\theta_+}\left( \Pi_{+}\ket{\psi} \right)\otimes \ket{1} \nonumber \\
    &+ e^{i\theta_-}\left( \Pi_{-}\ket{\psi} \right)\otimes \ket{-1}.
\end{align}
This defines the action of $\overline{\mathrm{GCNOT}}$ on $\hil\otimes \ket{1}$,\footnote{$\hil\otimes \ket{1}$ and $\hil\otimes \text{span}(\ket{1})$ are identical, so the former is a space as well.} so we can complete it by defining it on the orthogonal complement, namely $\hil\otimes \ket{-1}$. Note
\begin{equation}
    \overline{\mathrm{GCNOT}}\left( \hil \otimes \ket{1} \right) = \hil_+ \otimes \ket{1} \oplus \hil_- \otimes \ket{-1},
\end{equation}
where $\hil_\pm$ are the $\pm 1$-eigenspaces of the Paulian operator $P$. As $\overline{\mathrm{GCNOT}}$ is unitary, 
\begin{equation}
    \overline{\mathrm{GCNOT}}\left( \hil \otimes \ket{1} \right) \perp \overline{\mathrm{GCNOT}}\left( \hil \otimes \ket{-1} \right),
\end{equation}
which suggests
\begin{equation}
    \overline{\mathrm{GCNOT}} \left( \hil \otimes \ket{-1} \right) = \hil_- \otimes \ket{1} \oplus \hil_+ \otimes \ket{-1}.
\end{equation}
$\overline{\mathrm{GCNOT}}|_{\hil \otimes \ket{-1}}$ can therefore be any unitary map from $\hil \otimes \ket{-1}$ to $\hil_- \otimes \ket{1} \oplus \hil_+ \otimes \ket{-1}$; a simple example is
\begin{align}
    \overline{\mathrm{GCNOT}} \left( \ket{\psi} \otimes \ket{-1} \right) := &\left( U_- \Pi_- \ket{\psi} \right) \otimes \ket{1} \nonumber \\
    &+ \left( U_+ \Pi_+ \ket{\psi} \right) \otimes \ket{-1},\label{eq:gcnotex}
\end{align}
where $U_\pm$ are any unitary maps from $\hil_\pm$ to themselves (i.e., operators); we may also choose
\begin{align}
    \overline{\mathrm{GCNOT}} \left( \ket{\psi} \otimes \ket{-1} \right) := &\left( U_+' \Pi_+ \ket{\psi} \right) \otimes \ket{1} \nonumber \\
    &+  \left( U_-' \Pi_- \ket{\psi} \right) \otimes \ket{-1},
\end{align}
where $U_\pm'$ are any unitary maps from $\hil_\pm$ to $\hil_\mp$. \eqref{eq:gcnot} is a special case of \eqref{eq:gcnotex} with $U_\pm$ being identities and $\theta_\pm$ of \eqref{eq:gcnotge} both being $0$.

\section{Code parameters and codewords of a concatenated binary code} \label{sec:cococo}

Because $\hil_+$ is $\hil_\mathrm{in}^{\otimes q}$ and there are $q(n_\text{in}-k_\text{in})$ of $Z_{i,j}^+$'s and $n_\text{out}-k_\text{out}$ of $Z_{i}^+$'s,
\begin{align}
    n_+ &= n_\text{in} q,\\
    k_+ &= n_\text{in} q - q(n_\text{in}-k_\text{in}) - \left(n_\text{out}-k_\text{out}\right) = k_\text{out}.
\end{align}
When $k_\text{out} = k_\text{in} = 1$, $n_+ = n_\text{in} n_\text{out}$ and $k_+ = 1$, as expected \cite{Preskill,Gottesman10}. That $k_+ = k_\mathrm{out}$ is also obvious from the way the codewords are obtained, as will be shown below. 

To find the distance, in the simple case of $k_\text{in}=1$, to change the logical state of the concatenated code an operator has to act nontrivially on at least $d_\text{out}$ inner logical qubits; for the logical state of the inner code to change, the operator needs to act nontrivially on at least $d_\text{in}$ physical qubits of $\hil_\text{in}$, so the distance is at least $d_\text{out} d_\text{in}$ \cite{Gottesman97,Gottesman10,Grassl09}.

More generally:
\begin{enumerate}
    \item To change the logical state of the concatenated code at least $d_\text{out}$ inner logical qubits should be acted upon nontrivially.
    \item Each $\hil_\text{in}$ subsystem contains $k_\text{in}$ inner logical qubits.
    \item To change the logical state of an $\hil_\text{in}$ system, i.e. the state of its logical qubits, at least $d_\text{in}$ physical qubits need to be acted on nontrivially.
\end{enumerate}
Therefore, the distance of the concatenated code satisfies \cite{Gottesman97,Gottesman10,Grassl09}
\begin{equation}
    d_+ \geq \lceil d_\text{out}/k_\text{in}\rceil d_\text{in}.
\end{equation} 

We can obtain the codewords given those of the outer and inner codes. For example, say that $(\ket{1,-1,-1,1} + \ket{-1,-1,-1,-1})/\sqrt{2}$ is a codeword of the outer code, and the inner code has two logical qubits. This codeword will become
$$ \left( \overline{\ket{1,-1}} \otimes  \overline{\ket{-1,1}} + \overline{\ket{-1,-1}} \otimes \overline{\ket{-1,-1}} \right)/\sqrt{2},$$
where $\overline{\ket{i,j}}\in\hil_\text{in,C}$ are logical states of the inner code, $i$ for the first logical qubit and $j$ for the second one. 

So far we have taken the outer code as composed of $n_\text{out}$ qubits, but we can treat it as composed of $q$ subsystems each with dimension $\dim \hil_\text{in,C}$ instead. We can then follow the standard procedure for finding the codewords and parameters of a concatenated code as in, e.g., Refs.~\cite{Knill96,Gottesman97,Grassl09,Wang13}: Replacing each $\dim \hil_\text{in,C}$-dimensional subsystem of the outer code by $\hil_\mathrm{in}$, the concatenated code hence has $\dim \hil_+ = (2^{n_\text{in}})^{q}$ and $\dim \hil_{+,\mathrm{C}} = 2^{k_\text{out}}$. The distance is no smaller than the product of that of the outer code and that of the inner one; note as compared to when it is regarded as composed of qubits, the distance of the outer code now should be reduced by a factor of $k_\mathrm{in}$ because each of its $q$ subsystems is composed of $k_\mathrm{in}$ qubits.

\section{Phaseless group}
\label{sec:phaseless}

For a group of operators $\mathsf{G}$ containing $\{I,-I,iI,-iI\}$ as a subgroup, we define 
\begin{equation}
    \hat{\mathsf{G}} := \mathsf{G}/\{I,-I,iI,-iI\}\label{eq:quo},
\end{equation}
As $\{I,-I,iI,-iI\}$ is clearly normal, $\hat{\mathsf{G}}$ is a quotient group \cite{Lang}. Specifically, for the Pauli group $\hat{\mathsf{P}}^n$ can be regarded as a ``phaseless'' version of the Pauli group: Abusing the language, we will regard the coset representatives as elements of $\hat{\mathsf{P}}^n$; consequently, we will also call $\hat{\mathsf{P}}^n$ a Pauli group and its elements \emph{Pauli operators}. By removing the phases, $\hat{\mathsf{P}}^n$ becomes linearly independent; in particular, $\hat{\mathsf{P}}^n$ is a basis of $\lmap{\mathbb{C}^{2^n}}$ \cite{Gottesman97,Knill96_1,Knill96_2}. The phaseless Pauli group $\hat{\mathsf{P}}^1$ is isomorphic to the Klein four-group and $\hat{\mathsf{P}}^n$ is isomorphic to the direct product of $n$-copies of the Klein four-group \cite{Levick16}. 

For the same reason as why we introduced the phaseless Pauli group, given a subgroup of a Pauli/Paulian group, it is convenient to consider the phaseless version of it. As $\left\{ I,-I,iI,-iI \right\}$ may not always be in such a subgroup, we define: 
\begin{enumerate}
    \item If $\mathsf{G}$ has $\left\{ I,-I,iI,-iI \right\}$ as a subgroup, then its phaseless counterpart is defined like before, i.e., \eqref{eq:quo}. 
    \item If $iI \notin \mathsf{G}$ but $-I\in \mathsf{G}$, then 
    \begin{equation}
        \hat{\mathsf{G}} := \mathsf{G}/\{I,-I\}.
    \end{equation}
    \item  Finally, if $-I\notin \mathsf{G}$,
    \begin{equation}
        \hat{\mathsf{G}} := \mathsf{G}. \label{eq:hgg}
    \end{equation}
\end{enumerate} 
Like before, we take the coset representatives as elements of such a quotient group, which is the reason why we ``defined'' $\hat{\mathsf{G}}$ as $\mathsf{G}$ in \eqref{eq:hgg}. Correspondingly there is arbitrariness in the choice of its elements: Indeed, saying $P \in \hat{\mathsf{G}}$ is no different from saying $P\in \mathsf{G}$. It is only when we compare sets does it make a difference: For example, if we say a set $\mathbb{S}$ is equal to $\hat{\mathsf{G}}$, then there should exist no two elements in $\mathbb{S}$ that differ by a nontrivial multiplication factor, cf. \eqref{eq:hgg} and later \eqref{eq:hss}.

\section{Commutativity of Pauli subgroups}
\label{sec:commpa}

Here is a property concerning subgroups of Pauli groups:
\begin{lemma} \label{lem:pasub}
    For any subgroup $\mathsf{G}$ of a Pauli group, $-I \notin \mathsf{G}$ if and only if $\mathsf{G}$ is linearly independent, and only if $\mathsf{G}$ is composed wholly of involutions and is abelian.
\end{lemma}

\begin{proof}
    Suppose $-I\in \mathsf{G}$. Being a subgroup of the Pauli group, every element of $\mathsf{G}$ is either an involution or a counterinvolution. If $A\in \mathsf{G}$ were a counterinvolution, then $A^2 = -I$ would also be in $\mathsf{G}$, a contradiction, so every element is involutory. Next, a pair of Pauli operators either commute or anticommute. If $A,B\in \mathsf{G}$ anticommuted, $ABA^{-1} = - B = (- I) B\in \mathsf{G}$, so $ABA^{-1} B^{-1} = -I\in \mathsf{G}$, a contradiction. Hence every element in $\mathsf{G}$ commutes with one another, meaning $\mathsf{G}$ is abelian. Note that because $-I$ is an involution and commutes with all operators, $\mathsf{G}$ being abelian and comprising purely involutions does not imply $-I\in \mathsf{G}$

    Because $\left\{I,-I\right\}$ is linearly dependent, clearly a linearly independent subgroup should not contain $-I$. The other way around, assume $-I \notin \mathsf{G}$. Since the phaseless Paulian group (or the collection of its coset representatives) is a basis \cite{Gottesman97,Knill96_1,Knill96_2}, any subset of it is also linearly independent---In other words, any subset of a Pauli group, if no element differs from another by a multiplication factor, is linearly independent. As $(i I)^2 = -I$, $-I\notin \mathsf{G}$ implies that $i I$ and $-i I$ are not in $\mathsf{G}$ either; if $A\in \mathsf{G}$ and $aA \in \mathsf{G}$ for some nontrivial multiplication factor $a$ (namely $a$ is $-1$ or $i$ or $-i$), then $a I\in \mathsf{G}$ because $A^{-1} (aA) = a I$, a contradiction. Hence $\mathsf{G}$ is linearly independent.
\end{proof}

As $-I$ and $\pm i I$ commute with all operators, an abelian subgroup of $\mathsf{P}^n$ can contain the subgroup $\left\{  I, -I \right\}$ or $\left\{ I, -I , i I, -i I \right\}$, in which case the abelian subgroup is linearly dependent. To get rid of these extra factors, we can take the phaseless group of it, as we did in \eqref{eq:quo}. When we refer to a subgroup $\mathsf{S}$ of a Pauli, or Paulian, group as \textbf{maximally linearly independent and abelian}, it means that we cannot add any more Pauli or Paulian operators to it while keeping the subgroup both linearly independent and abelian; to put it another way, $\mathsf{S}$ is abelian and
\begin{equation}
    \hat{\mathsf{S}} = \mathsf{S}. \label{eq:hss}
\end{equation}
Some properties of such a group are revealed by Lemma~\ref{lem:pasub}; in particular this lemma shows that for a Pauli or Paulian\footnote{In this work a Paulian group is unitarily similar to a ``Pauli'' group (see Sec.~\ref{sec:stbgp}), so Lemma~\ref{lem:pasub} also holds for Paulian subgroups.} subgroup to be linearly independent, it must be abelian, so calling it abelian is actually redundant, but it helps show off this important attribute. 

\section{Proof for Lemma~\ref{lem:i+u}} \label{sec:prlem}

If: Let $A = a I + b U$ be an operator for which $a$ and $b$ are scalars and $U$ is a unitary operator. Since $I$ and $U$ commute, apparently $A^\dagger A = A A^\dagger$. Note this holds true whether $A$ is on $\mathbb{C}^2$ and whether $U$ is Paulian.

Only if: Let $A$ be a normal operator on $\mathbb{C}^2$. Because it is normal, it is unitarily diagonalizable and the eigenspaces are orthogonal. If there is only one eigenvalue, then $A$ is proportional to $I$; if it has two different eigenvalues $c_1$ and $c_2$, we can solve the equations $c_1 = a + b$ and $c_2 = a - b$. Suppose $A$ becomes diagonal under a unitary $V$, which in terms of a matrix means
\begin{equation}
    V A V^{-1} = \begin{pmatrix}
        a + b & 0 \\
        0 & a - b
    \end{pmatrix}.
\end{equation}
Consider the matrix of Pauli $Z$:
\begin{equation}
    Z = \begin{pmatrix}
        1 & 0 \\
        0 & -1
    \end{pmatrix};
\end{equation}
we have $V A V^{-1} = a I + b Z$, so $A = a I + b V^{-1} Z V$. Because $V^{-1} Z V$ is unitarily similar to $Z$, it has the same spectrum $\left\{ 1,-1 \right\}$ and is Paulian. 

\section{Codeword stabilized codes}
\label{sec:costb}

First, we will establish some properties of codeword stabilized codes that are essential in constructing Paulian stabilizers in Sec.~\ref{sec:codpre}, and then we will provide the steps to do so, and details of the two codes from Sec.~\ref{sec:csc}. In the following discussion we will assume the system is composed of $n$ qubits; please refer to Sec.~\ref{sec:csc} for the relevant terminologies and how codeword stabilized codes are constructed. 

In this section we will express Pauli operators in the following way, e.g., 
\begin{equation}
    XIZ = X\otimes I \otimes Z,
\end{equation}
and suppose the system is composed of three qubits:
\begin{equation}
    X_2 = I \otimes X \otimes I = IXI.
\end{equation}
In addition, $I$ in this subsection will refer to the identity operator on a single qubit. To distinguish the identity operator on the whole system from those on individual qubits, we will label the former as $I_n$, assuming the whole system is composed of $n$ qubits, so for example
\begin{equation}
    I_3 = III.
\end{equation}

\subsection{Preliminaries}
\label{sec:codpre}

Let $\mathsf{S}_w$ denote the word stabilizer. In general we will consider a particular generating set $g$ of $\mathsf{S}_w$, which will be taken as a tuple of generators, so we can associate each simultaneous eigenspace of $\mathsf{S}_w$ with a unique $n$-tuple of $\pm 1$ that is the simultaneous eigenvalues with respect to $g$, just like the error syndromes in relation to Paulian stabilizers. Such a tuple of simultaneous eigenvalues will be denoted by $\hat{t}$, and the set of all these tuples by $\mathbb{W}$. Like the tuples of syndromes $(t)$, we will use $\hat{t}$ to label spaces and the like, e.g. $\hil_{\hat{t}}$ is the $\hat{t}$-simultaneous eigenspace of $g$, which, to put another way, is a bijective map from $\mathbb{W}$ to the collection of all simultaneous eigenspaces of $\mathsf{S}_w$ or $g$:
\begin{equation}
    \mathbb{W} \ni \hat{t} \mapsto \hil_{\hat{t}},
\end{equation} 
cf. the syndrome map in Sec.~\ref{sec:stbgp}.

Let the state stabilized by the word stabilizer be $\ket{s}$, and $W$ be any Pauli operator in $\mathsf{P}^n$. Due to commutativity and anticommutativity, for any simultaneous eigenvector of a set or group of commutative Pauli operators, after acted upon by a Pauli operator it is still a simultaneous eigenvector of the same set or group of Pauli operators, so $W\ket{s}$ is also a simultaneous eigenvector of $\mathsf{S}_w$, which implies:
\begin{lemma}\label{lem:p1p2or}
    Given a codeword stabilized code, for two Pauli operators $P_1,P_2\in \mathsf{P}^n$, either $P_1 \ket{s} \propto P_2\ket{s}$ or $P_1 \ket{s} \perp P_2\ket{s}$. Hence, either $P_1 \hils{C} \perp P_2 \hils{C}$ or $P_1 \hils{C} \cap P_2 \hils{C} \neq \left\{ 0 \right\}$, i.e., $P_1 \hils{C} \perp P_2 \hils{C}$ if and only if $P_1 \hils{C} \cap P_2 \hils{C} = \left\{ 0 \right\}$.
\end{lemma}
\begin{proof}
    Since $P_1\ket{s}$ and $P_2\ket{s}$ are both simultaneous eigenvectors of $\mathsf{S}_w$, and because each simultaneous eigenspace is one-dimensional, they are either in the same eigenspace, i.e., proportional, or orthogonal. Likewise, as $\hils{C}$ is an orthogonal direct sum of simultaneous eigenspaces of $\mathsf{S}_w$, so are $P_1\hils{C}$ and $P_2\hils{C}$. Again, because the simultaneous eigenspaces of $\mathsf{S}_w$ are one-dimensional, either $P_1 \hils{C} \perp P_2 \hils{C}$ or $P_1 \hils{C} \cap P_2 \hils{C} \neq \left\{ 0 \right\}$.
\end{proof}

This lemma leads to
\begin{corollary}\label{cor:stbco}
    For a codeword stabilized code, a Pauli operator $P$ obeys $P|_\mathrm{C} \propto I_\mathrm{C}$ if and only if $P \in \mathsf{S}_w \left\{ I_n,-I_n,iI_n,-iI_n \right\}$ and $P$ commutes, or anticommutes, with all the word operators at the same time. In other words, a Pauli operator $P$ stabilizes the code if and only if 
    \begin{enumerate}[label={(\roman*)}] 
        \item either $P \in \mathsf{S}_w$ and $P$ commutes with all the word operators, \label{it:stbp}
        \item or $P\in -\mathsf{S}_w$ and $P$ anticommutes with all the word operators.
    \end{enumerate}
    Clearly, if $I_n$ is a word operator, the only possibility is $P$ commuting with all the word operators.
\end{corollary}
As we discussed in Sec.~\ref{sec:commpa}, because $\mathsf{S}_w$ is a maximal linearly independent and abelian subgroup of $\mathsf{P}^n$, $\mathsf{S}_w \left\{ I_n,-I_n,iI_n,-iI_n \right\}$ is a maximal abelian subgroup. A Pauli stabilizer for a codeword stabilizer code that is guaranteed to exist is $I$. As shown in Ref.~\cite{Cross09}, Pauli stabilizer codes are a special case of codeword stabilizer codes, which of course have nontrivial Pauli stabilizers. 
\begin{proof}
    Given a Pauli operator $W\in \mathsf{P}^n$, for $W\ket{s}$ to be an eigenvector of another Pauli operator $P$, $P$ must commute with all elements in $\mathsf{S}_w$ as $W\ket{s}$ is a simultaneous eigenvector of $\mathsf{S}_w$; because $\mathsf{S}_w$ is a maximal linearly independent and abelian subgroup of $\mathsf{P}^n$, it implies $P\in \mathsf{S}_w \left\{ I_n,-I_n,iI_n,-iI_n \right\}$. Furthermore, if $PW = \pm W P$, $PW\ket{s} = \pm WP\ket{s}$, which leads to the requirement on commutation relations between $P$ and the word operators. The other direction is pretty obvious and hence omitted.
\end{proof}

An error $E$ is detectable if and only if $E$ satisfies \cite{Preskill,Cross09}
    \begin{equation}
        \Pis{C}E\Pis{C} \propto \Pis{C}. \label{eq:pecp}
    \end{equation}
The following lemma considers a property of detectable errors:
\begin{lemma}\label{lem:a}
    Let $E$ be a unitary operator which obeys $\Pis{C}E\Pis{C} = a \Pis{C}$ for some scalar $a$; note $|a|\leq 1$ by unitarity of $E$. For such an operator we can find:
    \begin{enumerate}[label={(\roman*)}] 
        \item $|a|=0$ if and only if $E\hils{C} \perp \hils{C}$, in which case $E\hils{C}\cap \hils{C} = \left\{ 0 \right\}$. \label{con:a0}
        \item $|a|=1$ if and only if $E|_{\mathrm{C}} = e^{i\theta} I_{\mathrm{C}}$ for some real $\theta$, in which case $E\hils{C} = \hils{C}$.\label{con:a1}
        \item $|a|\in (0,1)$ if and only if $E\hils{C}$ and $\hils{C}$ are not orthogonal and $E\hils{C}\cap \hils{C} = \left\{ 0 \right\}$.\label{con:0a1}
    \end{enumerate}
    Correspondingly, the following conditions are equivalent:
    \begin{enumerate}[label={(\Roman*)}]        
        \item $|a| = 1$. \label{con:e1}
        \item $E|_{\mathrm{C}} = e^{i\theta} I_{\mathrm{C}}$ for some real $\theta$. \label{con:e2}
        \item $E\hils{C} = \hils{C}$. \label{con:e3}
        \item $E\hils{C} \cap \hils{C} \neq \left\{ 0 \right\}$. \label{con:e4}    
    \end{enumerate}
\end{lemma}
\begin{proof}
    \ref{con:a0} is obvious. 
    
    \ref{con:a1}: It is apparent that $E|_{\mathrm{C}} = e^{i\theta} I_{\mathrm{C}}$ implies $a = e^{i\theta}$ and thus $|a|=1$. To show the converse, we first note 
    \begin{equation}
        ||\Pis{C}\ket{w}|| \leq ||\ket{w}||
    \end{equation}
    which becomes an equality if and only if $\ket{w} \in \hils{C}$. Now, consider any $\ket{v}\in \hils{C}$ and let $a = e^{i\theta}$ we have
    \begin{align}
        ||\ket{v}|| & = ||\Pis{C}E\Pis{C}\ket{v}|| \nonumber \\
        &\leq || E \ket{v}|| = ||\ket{v}||,
    \end{align}
    implying $E\ket{v}\in \hils{C}$. As this is true for all $\ket{v} \in \hils{C}$ and $E$ is unitary, we have $E\hils{C} = \hils{C}$. Next, since $e^{i\theta}\ket{v} = \Pis{C}E \ket{v} = E\ket{v}$ for all $\ket{v} \in \hils{C}$, we obtain $E|_{\mathrm{C}} = e^{i\theta} I_{\mathrm{C}}$.
    
    \ref{con:0a1}: If $E\hils{C} \cap \hils{C} = \left\{ 0 \right\}$, for any nonzero $\ket{v} \in \hils{C}$, because $E\Pis{C}\ket{v} = E\ket{v} \notin \hils{C}$, 
    \begin{align}
        |a| \,||\ket{v}|| &= ||\Pis{C} E \Pis{C}\ket{v}|| \nonumber \\
        &< ||E\ket{v}|| \nonumber \\
        &= ||\ket{v}||,
    \end{align}
    implying $|a| < 1$, and because $E\hils{C}$ and $\hils{C}$ are not orthogonal, $|a|>0$. On the contrary, when $|a|\in (0,1)$, since $a\neq 0$, $E\hils{C}$ must not be orthogonal to $\hils{C}$. Should $E\hils{C}$ and $\hils{C}$ have nontrivial intersection, there exist nonzero $\ket{v} \in \hils{C}$ such that $E\ket{v}\in \hils{C}$, and for such $\ket{v}$, we have $||E \Pis{C} \ket{v}|| = ||E\ket{v}|| = ||\ket{v}||$, so
    \begin{equation}
        |a|\, ||\ket{v}|| = ||\Pis{C} E \Pis{C}\ket{v}|| = ||\ket{v}||,
    \end{equation}
    a contradiction; therefore $E\hils{C} \cap \hils{C} = \left\{ 0 \right\}$.

    Let's go on to show why conditions \ref{con:e1} to \ref{con:e4} are equivalent:
    \begin{itemize}
        \item By \ref{con:a1}, \ref{con:e1} $\rightarrow$ \ref{con:e2}, and \ref{con:e2} $\rightarrow$ \ref{con:e3}.
        \item Clearly \ref{con:e3} $\rightarrow$ \ref{con:e4}.
        \item Because only when $|a| = 1$ do $E\hils{C}$ and $\hils{C}$ intersect nontrivially, \ref{con:e4} implies \ref{con:e1}.
    \end{itemize}
    This completes the proof.
\end{proof}

\begin{corollary} \label{cor:codep}
    For a codeword stabilized code:
    \begin{enumerate}[label={(\alph*)}]  
        \item Every detectable Pauli error $P$ obeys either $P\hils{C}\perp \hils{C}$ or $P \in \mathsf{S}_w \left\{ I_n,-I_n,iI_n,-iI_n \right\}$; in the latter case $P|_\mathrm{C} \propto I_\mathrm{C}$ and it is hence also correctable. \label{it:p1}
        \item For every pair of correctable Pauli errors $P_1$ and $P_2$, either $P_1$ and $P_2$ are orthonormal or there exists $S \in \mathsf{S}_w \left\{ I_n,-I_n,iI_n,-iI_n \right\}$ such that $P_2 = P_1 S$; in the latter case $S|_\mathrm{C} \propto I_\mathrm{C}$ so $P_2|_\mathrm{C} \propto P_1|_\mathrm{C}.$ \label{it:p1p2}
    \end{enumerate}
\end{corollary}
Note that a correctable unitary operator is ``normalized'' according to our definition, so a correctable Pauli error is normalized.
\begin{proof}
    \ref{it:p1}: As the Pauli error $P$ is detectable, we will make use of Lemma~\ref{lem:a}: If $P\hils{C} = \hils{C}$, it obeys condition~\ref{con:a1} of Lemma~\ref{lem:a}, and according to Corollary~\ref{cor:stbco} it must be in $\mathsf{S}_w \left\{ I_n,-I_n,iI_n,-iI_n \right\}$; if not, it satisfies either \ref{con:a0} or \ref{con:0a1} of Lemma~\ref{lem:a}. As Lemma~\ref{lem:p1p2or} shows, \ref{con:0a1} cannot occur, so only \ref{con:a0} and \ref{con:a1} are possible, completing the proof for \ref{it:p1}. 
    
    \ref{it:p1p2}: For correctable Pauli operators $P_1$ and $P_2$, they satisfy $\Pis{C}P_1^\dagger P_2 \Pis{C} \propto \Pis{C}$. By Lemma~\ref{lem:p1p2or}, either $P_1\hils{C} \perp P_2\hils{C}$ or $P_1\hils{C} \cap P_2\hils{C} \neq \left\{ 0 \right\}$. If it is the former, $\Pis{C}P_1^\dagger P_2 \Pis{C} = 0$, i.e., they are orthonormal. If the latter, $\Pis{C}P_1^\dagger P_2 \Pis{C} \neq 0$, in which case we must have $P_1 \hils{C} = P_2 \hils{C}$ and they must act identically except for a multiplication factor on $\hils{C}$, else one could not invert the action of the other on $\hils{C}$. As $S:= P_1^{-1} P_2$ is also a Pauli operator and $S|_\mathrm{C}\propto I_\mathrm{C}$, $S\in\mathsf{S}_w \left\{ I_n,-I_n,iI_n,-iI_n \right\}$ by Corollary~\ref{cor:stbco}.
\end{proof}

Corollary~\ref{cor:csc} is \ref{it:deg} of the following corollary combined with \ref{it:p1} of Corollary~\ref{cor:codep}; below ``wt'' refers to the weight of an operator:
\begin{corollary}\label{cor:nonde}
    Consider an $n$-qubit codeword stabilized code with distance $d$, for which $\mathrm{wt}S \geq d$ for all $S\in \mathsf{S}_w\setminus \left\{ I_n \right\}$. 
    \begin{enumerate}[label={(\roman*)}]  
        \item If $P\in \mathsf{P}^n$ has $\mathrm{wt}P < d$ and is not proportional to $I_n$, then $I_n$ and $P$ are orthogonal. \label{it:pnoti}
        \item For $P_1,P_2,P_1 P_2\in \mathsf{P}^n$, suppose that their weights are all less than $d$ and that they are linearly independent; then elements in $\left\{ I_n, P_1,P_2,P_1 P_2 \right\}$ are mutually orthogonal. \label{it:ppid}
        \item For a pair of operators in $\mathsf{P}^n$, if their weights are both no higher than $\left\lfloor (d-1)/2 \right\rfloor$ and if they are not proportional to each other, they are orthonormal. \label{it:pap}
        \item The code is nondegenerate \cite{Gottesman97,Gottesman10}: Indeed, for a codeword stabilized code with distance $d$, it is nondegenerate if and only if $\mathrm{wt}S \geq d$ for all $S\in \mathsf{S}_w\setminus \left\{ I_n \right\}$. \label{it:deg}
    \end{enumerate}
\end{corollary}
Note this corollary does not imply that a codeword stabilized code whose word stabilizers obey the condition on weights above will have distance $d$---The code having distance $d$ is part of the assumption. Also, it may seem weird at first sight that word operators did not show up, even though they are essential in formulating a codeword stabilizer code. Their roles here are implicit: As we have assumed the code has distance $d$, Pauli errors of certain weights must obey specific conditions with the word stabilizer and word operators as listed in Ref.~\cite{Cross09}.
\begin{proof}
    \ref{it:pnoti}: Because $P$ is not in $\mathsf{S}_w \left\{ I_n,-I_n,iI_n,-iI_n \right\}$ (due to its weight) and is detectable, Corollary~\ref{cor:codep} implies $P$ and $I_n$ are orthogonal.

    \ref{it:ppid}: First, due to their weights and linear independence,\footnote{By linear independence none of them can be proportional to $I_n$.} $P_1,P_2,P_1 P_2$ are not in $\mathsf{S}_w \left\{ I_n,-I_n,iI_n,-iI_n \right\}$. Let $\mathsf{G}$ be the group generated by $P_1$, $P_2$, and their adjoints. By linear independence and the fact that the adjoint of a Pauli operator differs at most by a multiplication factor, we have
    \begin{equation}
        \hat{\mathsf{G}} = \left\{ I_n, P_1,P_2,P_1 P_2 \right\}.
    \end{equation}
    By Corollary~\ref{cor:codep} or \ref{it:pnoti}, for all $O\in \hat{\mathsf{G}}$ except $I_n$, $\Pis{C} O \Pis{C} = 0$, so for all $O_1,O_2\in \hat{\mathsf{G}}$ with $O_1\neq O_2$
    \begin{equation}
        \Pis{C} O_1^\dagger O_2 \Pis{C} = 0, \label{eq:orthooo}
    \end{equation}
    because $O_1^\dagger O_2$ is also in $\mathsf{G}$ and is not proportional to $I_n$; $O_1$ and $O_2$ are therefore orthogonal.
    
    \ref{it:pap}: Suppose we have $P_1,P_2\in \mathsf{P}^n$ that are not proportional to each other and whose weights are no higher than $\left\lfloor (d-1)/2 \right\rfloor$. If one of them is proportional to $I_n$, then they are orthonormal by \ref{it:pnoti}; else, $P_1 P_2$ will not be proportional to $I_n$ and $\mathrm{wt} P_1 P_2 < d$, so by \ref{it:ppid} $P_1$ and $P_2$ are orthonormal.

    \ref{it:deg}: By \ref{it:pap}, the condition on weights is a sufficient condition for nondegeneracy. To show it is a necessary condition, suppose there exists $S\in \mathsf{S}_w\setminus \left\{ I_n \right\}$ whose weight is less than $d$. From its weight, $S$ is detectable, and because it is in $\mathsf{S}_w$, it must act like an identity on the code space (according to Corollary~\ref{cor:codep}). Then there would exist nontrivial Pauli operators $P_1\neq P_2$ for which $\mathrm{wt} P_i \leq \left\lfloor (d-1)/2 \right\rfloor$, $i=1,2$ such that $P_1 = P_2 S$, so $P_1$ and $P_2$ act the same on the code space, implying the code is degenerate.
\end{proof}

\subsection{Constructing Paulian stabilizers}
\label{sec:conpastb}

In this part we will discuss how to construct Paulian stabilizers for codeword stabilized codes, where we will make heavy use of $\hil_{\hat{t}}$'s, i.e., the simultaneous eigenspaces of $\mathsf{S}_w$. Unlike the proof for Proposition~\ref{pro:main}, we will not attempt to build the minimal group first and expand upon it. A quick reminder: $g$ denotes a tuple of generators of $\mathsf{S}_w$ and $\mathbb{W}$ is the collection of all the tuples of simultaneous eigenvalues with respect to $g$; please refer back to the start of Sec.~\ref{sec:codpre}. Let's lay out the procedure: 
\begin{enumerate}[label={A.\arabic*}] 
    \item \label{it:a1} First, check if \eqref{eq:m'+1} is satisfied to see whether it is possible to correct all errors with Paulian stabilizers---Here we will assume this is true; we then choose a set $\mathbb{F}$ of orthonormal correctable Pauli errors. For a code with a given distance, we can use \eqref{eq:fcos2}, and select linearly independent Pauli errors with weight no higher than $\left\lfloor (d-1)/2 \right\rfloor$ as orthonormal correctable errors; specifically, if the code is nondegenerate, which can be easily checked via Corollary~\ref{cor:nonde}, we can choose all of them. 
    \item \label{it:a2} As discussed in the proof for Lemma~\ref{lem:p1p2or}, $P\hils{C}$ is a direct sum of simultaneous eigenspaces of $\mathbb{S}_w$ for any Pauli operator $P$; hence for $F\in \mathbb{F}$, let $\mathbb{W}_F$ denote the subset of $\mathbb{W}$ such that
    \begin{equation}
        F \hils{C} = \bigoplus_{\hat{t} \in \mathbb{W}_F} \hil_{\hat{t}}.\label{eq:wf}
    \end{equation}
    $\mathbb{W}_F$ can be found out by making use of the commutation relations between the word operators and $g$, and those between the word operators and $F$. 
    \item \label{it:a3} Since  
    \begin{equation}
        \bigoplus_{F\in \mathbb{F}} F \hils{C} = \bigoplus_{F \in \mathbb{F}} \bigoplus_{\hat{t} \in \mathbb{W}_F} \hil_{\hat{t}},
    \end{equation}
    and since for all $F_1,F_2\in \mathbb{F}$
    \begin{equation}
        \mathbb{W}_{F_1} \cap \mathbb{W}_{F_2} = \varnothing \;\text{if } F_1\neq F_2,
    \end{equation}
    where $\varnothing$ refers to the empty set, with
    \begin{equation}
        \mathbb{W}_\perp := \mathbb{W}\setminus \bigcup_{F\in \mathbb{F}} \mathbb{W}_F, \label{eq:wperp}
    \end{equation}
    we have
    \begin{equation}
        \left( \bigoplus_{F\in \mathbb{F}} F \hils{C} \right)^\perp = \bigoplus_{\hat{t} \in \mathbb{W}_\perp} \hil_{\hat{t}}.
    \end{equation} 
    The set $\mathbb{W}_\perp$ and the associated simultaneous eigenspaces will be used as ``spares.'' 
    \item \label{it:a4} Let $m = \left \lceil \log_2 \left|\mathbb{F}\right| \right \rceil$, and $\mathbb{T}$, as before, be the collection of all $m$-tuples of $\pm 1$, i.e., syndromes. Choose a unique syndrome for each error in $\mathbb{F}$, namely, a one-to-one map $f_\mathrm{sym}:\mathbb{F} \rightarrow \mathbb{T}$, the ``syndrome map,''\footnote{The syndrome map in Sec.~\ref{sec:stbgp} was defined on $\mathbb{F}'$ instead of $\mathbb{F}$, so it was bijective besides injective.} and we require $f_\mathrm{sym}(I_n) = (I)$, where $(I)$ is the tuple whose components are all 1. If $m > \log_2|\mathbb{F}|$, there will be ``excess'' syndromes that do not correspond to correctable errors; i.e., they are members of 
    \begin{equation}
        \mathbb{T} \setminus f_\mathrm{sym}(\mathbb{F}).
    \end{equation}
    The total number of excess syndromes is
    \begin{equation}
        \left|\mathbb{T} \setminus f_\mathrm{sym}(\mathbb{F}) \right| = \left|\mathbb{T} \right| - \left|f_\mathrm{sym}(\mathbb{F}) \right| = 2^m - |\mathbb{F}|.
    \end{equation}
    If $m = \log_2 \left|\mathbb{F}\right|$, then this already gives us the ``minimal'' Paulian stabilizers; see Sec.~\ref{sec:stbgp} and \ref{it:a6} on how to define Paulian stabilizers given the syndrome spaces. 
    \item \label{it:a5} Now let's designate all the syndrome spaces. The properties we desire of them are (cf. Sec.~\ref{sec:extdo}):
    \begin{enumerate}
        \item The syndrome space associated with each error $F \in \mathbb{F}$ should contain $F\hils{C}$: 
        \begin{equation}
            F \hils{C} \subseteq \hil_{f_\mathrm{sym}(F)} \; \forall F \in \mathbb{F}. \label{eq:fsymf}
        \end{equation}
        \item The syndrome spaces are orthogonal: For any two distinct syndromes $(s)$ and $(t)$,
        \begin{equation}
            \hil_{(s)} \perp \hil_{(t)}. \label{eq:sperpt}
        \end{equation}
        \item All syndrome spaces are isomorphic:
        \begin{equation}
            \hil_{(s)} \cong \hil_{(t)} \; \forall (s),(t) \in \mathbb{T}. \label{eq:scongt}
        \end{equation}
    \end{enumerate}
    To achieve them, for every $(t)\in \mathbb{T}$ we choose a subset $\mathbb{W}_{(t)}$ of $\mathbb{W}$ and demand the syndrome spaces be
    \begin{equation}
        \hil_{(t)} := \bigoplus_{\hat{s} \in \mathbb{W}_{(t)}} \hil_{\hat{s}};
    \end{equation}
    $\mathbb{W}_{(t)}$'s shall satisfy the following conditions:
    \begin{enumerate}
        \item To comply with \eqref{eq:fsymf}, for all $(t)\in f_\mathrm{sym}(\mathbb{F})$ we require
        \begin{equation}
            \mathbb{W}_{F_{(t)}} \subseteq \mathbb{W}_{(t)}; \label{eq:incw}
        \end{equation}
        see the definition of $\mathbb{W}_F$ for all $F\in \mathbb{F}$ in \eqref{eq:wf}.
        \item To satisfy \eqref{eq:sperpt},
        \begin{equation}
            \mathbb{W}_{(s)} \cap \mathbb{W}_{(t)} = \varnothing \;\forall (s)\neq (t). \label{eq:noninterw}
        \end{equation}
        \item To obey \eqref{eq:scongt},
        \begin{equation}
            |\mathbb{W}_{(s)}| = |\mathbb{W}_{(t)}| \;\forall (s),(t) \in \mathbb{T}. \label{eq:stcar}
        \end{equation}
        Since $\dim \hil_{\hat{t}} = 1$, $|\mathbb{W}_{(t)}|$ is the dimension of each syndrome space, and $|\mathbb{W}_{(t)}| - \dim \hils{C}$ can show how much we extend the domain of the Paulian stabilizers.
    \end{enumerate}
    Note 
    \begin{align}
        \mathbb{W}_{(t)} \setminus \mathbb{W}_{F_{(t)}} &\subseteq \mathbb{W}_\perp \; \forall (t)\in f_\mathrm{sym}(\mathbb{F}), \\
        \mathbb{W}_{(t)} &\subseteq \mathbb{W}_\perp \; \forall (t)\in \mathbb{T} \setminus f_\mathrm{sym}(\mathbb{F}),
    \end{align} 
    so the spares---$\mathbb{W}_\perp$ of \eqref{eq:wperp} and the associated eigenspaces---are used to fill in each syndrome space. Finally, if the stabilizers are to be Paulian on the entire space $\hil$, the dimension of each syndrome space is
    \begin{equation}
        \dim\hil_{(t)} = 2^n / 2^m = 2^{n-m}, 
    \end{equation}
    i.e., each is composed of $n-m$ qubits. 
    \item \label{it:a6} With the syndrome spaces specified, we have the corresponding Paulian stabilizers: For all $i = 1,\dotsc,m$,
    \begin{equation}
        Z_{i}^\mathrm{S} := \Pi_{\bigoplus_{(t)_i = 1,(t)\in \mathbb{T},} \hil_{(t)}} - \Pi_{\bigoplus_{(t)_i = -1, (t)\in \mathbb{T}} \hil_{(t)}},
    \end{equation}
    where $(t)_i$ is the $i$-th component of $(t)$. The domain of the Paulian stabilizers, $\hil'$ of Proposition~\ref{pro:main}, is thus
    \begin{equation}
        \hil' = \bigoplus_{(t)\in \mathbb{T}} \hil_{(t)}.
    \end{equation}
    Defined this way, each $Z_{i}^\mathrm{S}$ is clearly Paulian to the restriction of $\hil'$. They commute, with $(t)$-simultaneous eigenspaces $\hil_{(t)}$; i.e., $(t)$'s are the error syndromes and $\hil_{(t)}$'s are the corresponding syndrome spaces. Because we have demanded $I_n$ have syndrome $(I)$, $Z_{i}^\mathrm{S}$'s are stabilizers. 
\end{enumerate}
In the examples to come, we will demonstrate how to put them into practice.

\subsection{\texorpdfstring{$((5,6,2))$}{((5,6,2))}-code} \label{sec:562}
Let's start off with the $((5,6,2))$-code from Refs.~\cite{Rains97,Cross09}. As discussed in Sec.~\ref{sec:csc}, it is impossible to find a Paulian stabilizer group that can detect all the weight-1 errors for this code, but due to its low dimensions, it is easier to demonstrate the procedure shown in Sec.~\ref{sec:conpastb} with this code, and we can also show how to adapt the methods for error-detecting codes. 

The word stabilizer of this code is generated by $ZXZII$ and all its cyclic shifts, i.e.,
\begin{align} 
    g = ( & ZXZII, XZIIZ, ZIIZX, \nonumber \\
    &IIZXZ, IZXZI ), \label{eq:sw62}
\end{align}
and the word operators are 
\begin{align}
    &IIIII, ZZIZI, IZZIZ,  \nonumber \\
    &ZIZZI, IZIZZ, ZIZIZ.
\end{align}
Now let's follow the steps listed in Sec.~\ref{sec:conpastb}:
\begin{itemize}
    \item \ref{it:a1}: As this code is an error-detecting code, how do we choose orthonormal Pauli errors? In fact, because the code has distance 2, we infer from \ref{it:ppid} of Corollary~\ref{cor:nonde} that $\mathbb{F} = \left\{ X_i,Y_i,Z_i,I_5 \right\}$ is orthonormal for a fixed $i=1,\dotsc,5$, and we will use them as the orthonormal ``correctable'' errors.
    \item \ref{it:a2}: We should find $\mathbb{W}_F$ for each $F\in \mathbb{F}$. As a demonstration, we will show how to find $\mathbb{W}_{X_1}$. First, consider the word operator $ZZIZI$. Its commutation relation with $g$ of \eqref{eq:sw62}, if expressed as a tuple of $\pm 1$ with $+1$ for commuting and $-1$ for anticommuting, is 
    \begin{equation}
        (-1,-1,1,-1,1), \label{eq:1-1t1}
    \end{equation}
    which is exactly the simultaneous eigenvalues of the vector $ZZIZI\ket{s}$ with respect to $g$. Repeating for all the word operators, we obtain $\hils{C}$ as the direct sum of simultaneous eigenspaces of $g$ or $\mathsf{S}_w$. To obtain $\mathbb{W}_{X_1}$, we check how $X_1$ commutes with $g$: The commutation relation is
    \begin{equation}
        (-1,1,-1,1,1),\label{eq:1-1t2}
    \end{equation}
    which means $X_1 (ZZIZI\ket{s})$ has simultaneous eigenvalues 
    \begin{align}
        &\left(-1 \times (-1),-1 \times 1,1 \times (-1),-1 \times 1,1 \times 1\right) \nonumber \\
        = & (1,-1,-1,-1,1), \label{eq:1-1t3}
    \end{align}
    namely multiplications component by component between \eqref{eq:1-1t1} and \eqref{eq:1-1t2}; $(1,-1,-1,-1,1)$ from \eqref{eq:1-1t3} is therefore an element of $\mathbb{W}_{X_1}$. Doing this all over again for all the word operators gives us $\mathbb{W}_{X_1}$.
    \item \ref{it:a3}: After obtaining each $\mathbb{W}_F$, $\mathbb{W}_\perp$ should have $2^5 - \dim\hils{C} \times 4 = 8$ elements.
    When $i=1$, they are
    \begin{align}
        \hat{a} &:= (-1, -1, -1, 1, -1),\nonumber \\
        \hat{b} &:= (-1, -1, -1, 1, 1),\nonumber \\
        \hat{c} &:= (-1, 1, 1, 1, -1),\nonumber \\
        \hat{d} &:= (1, -1, -1, -1, -1),\nonumber \\
        \hat{e} &:= (-1, 1, 1, 1, 1),\nonumber \\
        \hat{f} &:= (1, -1, -1, -1, 1),\nonumber \\
        \hat{g} &:= (1, 1, 1, -1, -1),\nonumber \\
        \hat{h} &:= (1, 1, 1, -1, 1).
    \end{align}
    \item \ref{it:a4}: Since $\log_2 |\mathbb{F}| = 2$ is an integer, in this case we do not have any excess syndromes. Let's choose the syndrome for each element of $\mathbb{F}$, e.g.
    \begin{align}
        &F_{(1,1)} = I_5, \;
        F_{(1,-1)} = X_i, \nonumber \\ 
        &F_{(-1,1)} = Y_i, \;
        F_{(-1,-1)} = Z_i;
    \end{align}
    a quick reminder: $(I) = (1,1)$ in this case. They give us the minimal Paulian stabilizers:
    \begin{align}
        Z_1^\mathrm{S}|_{\overline{\hil}} &= \Pi_{\hils{C} \oplus X_i\hils{C}} - \Pi_{Y_i\hils{C} \oplus Z_i\hils{C}}, \nonumber \\
        Z_2^\mathrm{S}|_{\overline{\hil}} &= \Pi_{\hils{C} \oplus Y_i\hils{C}} - \Pi_{X_i\hils{C} \oplus Z_i\hils{C}}. \label{eq:z1z2}
    \end{align}
    \item \ref{it:a5}: As addressed in the previous point, we have already had the minimal Paulian stabilizers, and we would like to extend their domain to the whole space while keeping them Paulian. We can choose
    \begin{align}
        \mathbb{W}_{(1,1)} \setminus \mathbb{W}_{F_{(1,1)}} &= \big\{ \hat{a},\hat{b} \big\}, \nonumber\\
        \mathbb{W}_{(1,-1)} \setminus \mathbb{W}_{F_{(1,-1)}} &= \big\{ \hat{c},\hat{d} \big\}, \nonumber\\
        \mathbb{W}_{(-1,1)} \setminus \mathbb{W}_{F_{(-1,1)}} &= \big\{ \hat{e},\hat{f} \big\}, \nonumber\\
        \mathbb{W}_{(-1,-1)} \setminus \mathbb{W}_{F_{(-1,-1)}} &= \big\{ \hat{g},\hat{h} \big\},
    \end{align}
    so
    \begin{align}
        \hil_{(1,1)} &= \hils{C} \oplus \hil_{\hat{a}} \oplus \hil_{\hat{b}}, \nonumber \\
        \hil_{(1,-1)} &= X_i\hils{C} \oplus \hil_{\hat{c}} \oplus \hil_{\hat{d}}, \nonumber \\
        \hil_{(-1,1)} &= Y_i\hils{C} \oplus \hil_{\hat{e}} \oplus \hil_{\hat{f}}, \nonumber \\
        \hil_{(-1,-1)} &= Z_i\hils{C} \oplus \hil_{\hat{g}} \oplus \hil_{\hat{h}}.
    \end{align}
    \item \ref{it:a6}: Now we have commutative stabilizers that are Paulian on the whole space:
    \begin{align}
        Z_1^\mathrm{S} &= \Pi_{\hil_{(1,1)} \oplus \hil_{(1,-1)} } - \Pi_{ \hil_{(-1,1)} \oplus \hil_{(-1,-1)} } \nonumber \\
        Z_2^\mathrm{S} &= \Pi_{\hil_{(1,1)} \oplus \hil_{(-1,1)} } - \Pi_{ \hil_{(1,-1)} \oplus \hil_{(-1,-1)} }.
    \end{align}
\end{itemize}

With $X_i$, $Y_i$, $Z_i$, and $I_5$ chosen as the orthonormal correctable errors, they have distinct syndromes with respect to the Paulian stabilizers, and we can correct their linear combinations, i.e., all errors occurring on the $i$-th qubit. As discussed earlier, the Paulian stabilizers for this code cannot detect all weight-1 errors; however, it can be found that with our choice of the syndrome spaces all single $X$ errors can be detected: Each single $X$ error maps $\hils{C}$ to a subspace of $\hil_{(I)}^\perp$, so the syndrome is different from $(I)$ and is detectable.

\subsection{\texorpdfstring{$((9,12,3))$}{((9,12,3))}-code}\label{sec:9123}
Now let's consider the $((9,12,3))$-code from Refs.~\cite{Yu08,Cross09}, which, unlike the previous example, is a legitimate error-correcting code. Since we have by and large demonstrated the methods in our previous example, we will only focus on the key points, and since the dimension is too large we will not give explicit forms of the Paulian stabilizers. 

\begin{enumerate}
    \item \ref{it:a1}: The word stabilizer is generated by $ZXZIIIIII$ and all its cyclic shifts, so it is apparent that
    \begin{equation}
        \mathrm{wt}S \geq d = 3 \; \forall S\in \mathsf{S}_W\setminus \left\{ I \right\}.
    \end{equation}
    By Corollary~\ref{cor:nonde} the code is nondegenerate, so we choose all linearly independent Pauli errors with weight no larger than $2$, which means by \eqref{eq:fcos} we have
    $$
        \left|\mathbb{F}\right| = 1 + 9 \times 3 = 28.
    $$
    Because 
    $$
        2^{\left \lceil \log_2 \left|\mathbb{F}\right| \right \rceil} \dim \hils{C} = 2^5 \times 12 < 2^5 \times 2^4 = \dim \hil = 2^9,
    $$
    it is possible for this code to have Paulian stabilizers that correct all the relevant errors. 
    \item \ref{it:a2} and \ref{it:a3} are routine. 
    \item \ref{it:a4}: As $m = \left\lceil \log_2 \left|\mathbb{F}\right| \right\rceil = 5 > \log_2 \left|\mathbb{F}\right|$, we will have excess syndromes in this case, and they are $2^m-\left|\mathbb{F}\right| = 4$ in total.
    \item \ref{it:a5}: If we want the stabilizers to be Paulian on the whole space, then each syndrome space is composed of $n-m = 4$ qubits. For each syndrome $(t)$ that points to an error $F_{(t)}$ in $\mathbb{F}$, $\dim \hil_{(t)} - \dim F_{(t)} \hils{C} = \dim \hil_{(t)} - \dim \hils{C} = 4$, so we need four elements of $\mathbb{W}_\perp$ to construct the associated syndrome space $\hil_{(t)}$, while for each excess syndrome we need $2^4 = 16$ elements of $\mathbb{W}_\perp$. 
    \item We can define the Paulian stabilizers following \ref{it:a6}. As there are four excess syndromes, there are four syndrome spaces no correctable errors will map the code space into---They exist to make the stabilizers Paulian.
\end{enumerate}

If we want to use the three Pauli stabilizers from Ref.~\cite{Yu08}, since they are also part of the word stabilizer (Corollary~\ref{cor:stbco}), it is better to let them be in the tuple of generators $g$, and steps \ref{it:a4} and \ref{it:a5} should be done accordingly; e.g., in \ref{it:a4} the syndrome for each orthonormal Pauli error should be chosen by how the error commutes with the Pauli stabilizers---so that these Pauli stabilizers will be among the Paulian stabilizers built in step \ref{it:a6}.

\section{Gottesman-Kitaev-Preskill codes} \label{sec:gkp}

Let $q = \left( a^\dagger+ a \right) / \sqrt{2}$ and $p = i\left( a^\dagger - a \right) / \sqrt{2}$ be conjugate quadrature operators. A Gottesman-Kitaev-Preskill (GKP) code for a single oscillator has two stabilizers, which are $e^{2i\pi q/\alpha}$ and $e^{-i n\alpha p}$ for some real $\alpha$, where $n$ is the dimension of the code space \cite{Gottesman01,Terhal16}; clearly the stabilizers are not Paulian. Such codes can correct small shifts in both $q$ and $p$; specifically, they can correct displacements with $\left| \Delta q \right| < \alpha /2$ and $\left| \Delta p \right| < \pi/ (n \alpha)$ \cite{Gottesman01}. The eigenstates of these stabilizers are not physical in that they are infinitely-squeezed, so in practice finitely-squeezed states are used; the error probability can be acceptably low if the state is squeezed sufficiently \cite{Gottesman01,Campagne20}. If the anticipated errors in $q$ and $p$ are comparable in magnitude, ``square'' GKP codes can be used, by choosing $\alpha = \sqrt{2\pi/n}$. When $n = \dim\hils{C} = 2$, the stabilizers are $e^{2i\sqrt{\pi} q}$ and $e^{-2i\sqrt{\pi} p}$ \cite{Gottesman01,Terhal20,Campagne20}.

To measure the syndrome, one way is by preparing the ancilla in a GKP state, and utilizing the Steane circuit to ascertain the amount of shifts by measuring the ancilla \cite{Gottesman01,Steane97,Vuillot19,Terhal20}. The outcomes are analog (or connected) rather than binary \cite{Terhal20}, and the corresponding measurement on the system is therefore not Paulian. Another avenue is phase estimation \cite{Terhal16,Terhal20}: Given a unitary operator $U$ on a system, if the system is in an $e^{i\theta}$-eigenstate, the procedure to estimate the phase, i.e., $\theta$ is called phase estimation. Because the stabilizers of GKP codes are unitary, we can obtain the syndrome this way; furthermore, as the simultaneous eigenspaces of $e^{2i\pi q/\alpha}$ and $e^{-i n\alpha p}$ are translations of the code space in $p$ and $q$, they are orthogonal and isomorphic \cite{Gottesman01}.

Phase estimation can be achieved by coupling the system and ancilla qubits via controlled-$U^k$ gates, and after performing suitable operations and measurements on the ancilla qubits we are able to approximate the phase $\theta$ \cite{Nielsen,Kitaev95,Griffiths96,Higgins07,Svore13,*Svore13arXiv}. It may seem that each measurement of an ancilla qubit is equivalent to measuring a Paulian operator on the system, as we have two measurement outcomes and they are equally likely (cf. Sec.~\ref{sec:pau}); however, it can be easily checked that such measurements in general are not orthogonal measurements, which is also evident from the coupling between the system and the ancilla being controlled-$U^k$, cf. Sec.~\ref{sec:cnot}---Hence, we cannot describe each measurement with a single self-adjoint operator, let alone a Paulian operator. 

Theoretically, we can construct commutative ``Paulian'' operators $Z_j^\mathrm{S}$'s for phase estimation: For convenience, let's rescale $\theta$, so that the eigenvalues of $U$ are $e^{2i\pi \theta}$ with $\theta \in [0,1)$ \cite{Nielsen}. Each $Z_j^\mathrm{S}$ is to measure the $2^{-j}$ digit of $\theta$ in binary representation, and $\theta = 0$ would correspond to the $(1,1,\cdots)$-simultaneous eigenvalues of $Z_j^\mathrm{S}$'s. Hence, with $\hil_{\theta}$ denoting the $e^{2i\pi\theta}$-eigenspace of $U$, we let the $1$ and $-1$-eigenspaces of $Z_j^\mathrm{S}$ be the direct sums of $\hil_\theta$ over all $\theta$ whose $2^{-j}$ digit in binary representation are $0$ and $1$, respectively. Under this construction, $Z_j^\mathrm{S}$'s shall be commutative and stabilize $1$-eigenvectors of $U$ (i.e., $\theta = 0$), and we can measure $Z_j^\mathrm{S}$'s to estimate the phase: For example, for an eigenvector of $U$ with $\theta = 0.1010$ in binary representation, it is a $(0,1,0)$-simultaneous eigenvector of $Z_j^\mathrm{S}$'s for $j = 1,2,3$. However, whether they are truly Paulian or not (as defined in Sec.~\ref{sec:pau}) depends on the spectral structure of $U$---The $\pm 1$-eigenspaces may fail to be isomorphic.

For GKP codes, we can construct phase estimation operators for $e^{2i\pi q/\alpha}$ and $e^{-i n\alpha p}$ respectively according to the previous paragraph, and these phase estimation operators are truly Paulian. The issue is that, even though they exist, to carry them out we need to couple very specific intervals of $\theta$ with the ancilla; see the $\pm 1$-eigenspaces of each $Z_j^\mathrm{S}$ above and Sec.~\ref{sec:cnot}. Hence, existing schemes for error correction, such as those in Refs.~\cite{Terhal16,Vuillot19,Terhal20,Campagne20,Hastrup21}, are more practical.

A closing remark: As discussed in Sec.~\ref{sec:dis}, commutative Paulian stabilizers are not unique, nor are the ones shown above. However, to construct practical Paulian stabilizers, appropriate syndrome spaces should be chosen, and this poses a great challenge, especially given the ``continuous'' nature of the errors for GKP codes. That being said, in practice states that approximate the true GKP codewords are used, and if confined to these physical states, we might be able to find suitable syndrome spaces to build practical Paulian stabilizers. This is, however, beyond the scope of this work.

% \begin{align}
%     D_q(\gamma) &:= \exp \left( i \frac{\sqrt{\pi}}{2}(2\gamma - 1) q\right), \nonumber \\
%     D_p(\gamma) &:= \exp \left( i \frac{\sqrt{\pi}}{2}(2\gamma - 1) p \right),
% \end{align}
% so these two operators (and their product) are correctable, and we will consider the binary representation of $\gamma$. We will define the Paulian operators as follows: Let the $+1$ and $-1$-eigenspaces of $S^1_{q}$ be the direct sum of all $D_q(\gamma) \hils{C}$ for which the $2^{-1}$-digit of $\gamma$ is 1 and 0 respectively; likewise for $S^1_{p}$. For all $m > 1$, we instead let $+1$ and $-1$-eigenspaces of $S^m_{q}$ be the direct sum of all $D_q(\gamma) \hils{C}$ for which the $2^{-1}$-digit of $\gamma$ is 0 and 1 respectively, similarly for $S^m_{p}$. By \eqref{eq:dhc}, $S^m_{q}$'s and $S^{m'}_{p}$'s are commutative Paulian operators; furthermore, because 
% \begin{equation}
%     \hils{C} = D_q(0.100\dots) \hils{C} = D_p(0.100\dots) \hils{C}
% \end{equation}
% in binary representation, $\hils{C}$ is stabilized by these operators. 

% To show what they entail, for $S_q^1$ its $+1$-eigenspace corresponds to a union of intervals in the $q$-space: 
% \begin{equation}
%     \bigcup_{n = -\infty}^{\infty} \left[2\sqrt{2\pi} n , \frac{\sqrt{\pi}}{2\sqrt{2}}  + 2\sqrt{2\pi} n\right),
% \end{equation}
% and the $-1$-eigenspace to
% \begin{equation}
%     \bigcup_{n = -\infty}^{\infty} \left[\frac{\sqrt{\pi}}{2\sqrt{2}} + 2\sqrt{2\pi} n , \frac{\sqrt{\pi}}{\sqrt{2}}  + 2\sqrt{2\pi} n\right).
% \end{equation}
% If a shift in $q$ occurs, 

\end{appendix}

\bibliographystyle{apsrev4-2}
\bibliography{bib}
\end{document}